\PassOptionsToPackage{unicode}{hyperref}
\PassOptionsToPackage{hyphens}{url}
\PassOptionsToPackage{dvipsnames,svgnames,x11names}{xcolor}

\documentclass[12pt]{article}
\usepackage{mathtools}
\usepackage{amsmath, amssymb, amsthm}
\usepackage{natbib}
\usepackage{subfiles}
\usepackage{xr} 
\usepackage{hyperref} 
\usepackage[affil-it]{authblk} 
\usepackage{placeins}
\usepackage{color}
\usepackage{bm}
\usepackage{caption}
\usepackage{subcaption}
\usepackage{gensymb}
\usepackage{adjustbox}
\usepackage{multirow}
\usepackage{booktabs}
\usepackage{graphicx} 
\usepackage{array}
\usepackage{float}
\usepackage{iftex}
\ifPDFTeX
  \usepackage[T1]{fontenc}
  \usepackage[utf8]{inputenc}
  \usepackage{textcomp} 
\else 
  \usepackage{unicode-math}
  \defaultfontfeatures{Scale=MatchLowercase}
  \defaultfontfeatures[\rmfamily]{Ligatures=TeX,Scale=1}
\fi
\usepackage{lmodern}
\ifPDFTeX\else  
\fi
\IfFileExists{upquote.sty}{\usepackage{upquote}}{}
\IfFileExists{microtype.sty}{
  \usepackage[]{microtype}
  \UseMicrotypeSet[protrusion]{basicmath} 
}{}
\makeatletter
\@ifundefined{KOMAClassName}{
  \IfFileExists{parskip.sty}{%
    \usepackage{parskip}
  }{
    \setlength{\parindent}{0pt}
    \setlength{\parskip}{6pt plus 2pt minus 1pt}}
}{
  \KOMAoptions{parskip=half}}
\makeatother
\usepackage{xcolor}
\makeatletter
\ifx\paragraph\undefined\else
  \let\oldparagraph\paragraph
  \renewcommand{\paragraph}{
    \@ifstar
      \xxxParagraphStar
      \xxxParagraphNoStar
  }
  \newcommand{\xxxParagraphStar}[1]{\oldparagraph*{#1}\mbox{}}
  \newcommand{\xxxParagraphNoStar}[1]{\oldparagraph{#1}\mbox{}}
\fi
\ifx\subparagraph\undefined\else
  \let\oldsubparagraph\subparagraph
  \renewcommand{\subparagraph}{
    \@ifstar
      \xxxSubParagraphStar
      \xxxSubParagraphNoStar
  }
  \newcommand{\xxxSubParagraphStar}[1]{\oldsubparagraph*{#1}\mbox{}}
  \newcommand{\xxxSubParagraphNoStar}[1]{\oldsubparagraph{#1}\mbox{}}
\fi
\makeatother

\usepackage{longtable,booktabs,array}
\usepackage{calc} 
\usepackage{etoolbox}
\makeatletter
\patchcmd\longtable{\par}{\if@noskipsec\mbox{}\fi\par}{}{}
\makeatother
\IfFileExists{footnotehyper.sty}{\usepackage{footnotehyper}}{\usepackage{footnote}}
\makesavenoteenv{longtable}
\usepackage{graphicx}
\makeatletter
\def\maxwidth{\ifdim\Gin@nat@width>\linewidth\linewidth\else\Gin@nat@width\fi}
\def\maxheight{\ifdim\Gin@nat@height>\textheight\textheight\else\Gin@nat@height\fi}
\makeatother
\setkeys{Gin}{width=\maxwidth,height=\maxheight,keepaspectratio}
\makeatletter
\def\fps@figure{htbp}
\makeatother

\makeatletter
\@ifpackageloaded{caption}{}{\usepackage{caption}}
\AtBeginDocument{%
\ifdefined\contentsname
  \renewcommand*\contentsname{Table of contents}
\else
  \newcommand\contentsname{Table of contents}
\fi
\ifdefined\listfigurename
  \renewcommand*\listfigurename{List of Figures}
\else
  \newcommand\listfigurename{List of Figures}
\fi
\ifdefined\listtablename
  \renewcommand*\listtablename{List of Tables}
\else
  \newcommand\listtablename{List of Tables}
\fi
\ifdefined\figurename
  \renewcommand*\figurename{Figure}
\else
  \newcommand\figurename{Figure}
\fi
\ifdefined\tablename
  \renewcommand*\tablename{Table}
\else
  \newcommand\tablename{Table}
\fi
}
\@ifpackageloaded{float}{}{\usepackage{float}}
\floatstyle{ruled}
\@ifundefined{c@chapter}{\newfloat{codelisting}{h}{lop}}{\newfloat{codelisting}{h}{lop}[chapter]}
\floatname{codelisting}{Listing}

\makeatother
\makeatletter
\@ifpackageloaded{caption}{}{\usepackage{caption}}
\@ifpackageloaded{subcaption}{}{\usepackage{subcaption}}
\makeatother

\ifLuaTeX
  \usepackage{selnolig}  
\fi
\usepackage{bookmark}

\IfFileExists{xurl.sty}{\usepackage{xurl}}{} 
\hypersetup{
  pdftitle={Title},
  pdfauthor={Author 1; Author 2},
  pdfkeywords={3 to 6 keywords, that do not appear in the title},
  colorlinks=true,
  linkcolor={blue},
  filecolor={Maroon},
  citecolor={Blue},
  urlcolor={Blue},
  pdfcreator={LaTeX via pandoc}}

\theoremstyle{plain}

\newtheorem{proposition}{Proposition}

\theoremstyle{definition}
\newtheorem{assumption}{Assumption}
\newtheorem{remark}{Remark}

\newcommand{\anon}{1}

\begin{document}

\def\spacingset#1{\renewcommand{\baselinestretch}%
{#1}\small\normalsize} \spacingset{1}


\if1\anon
{ 
  \title{\bf A Bayesian Weakest-Link Framework for Joint Estimation of Material Strength and Stress Profile}
  \author{Shiyu He\thanks{Corresponding author:
    s35he@uwaterloo.ca}\hspace{.2cm}\\
    Department of Statistics and Actuarial Science, University of Waterloo\\
    and \\
    Samuel W.K. Wong\\
    Department of Statistics and Actuarial Science, University of Waterloo}
  \maketitle
} \fi

\if0\anon
{
  \bigskip
  \bigskip
  \bigskip
  \begin{center}
    {\LARGE\bf A Bayesian Weakest-Link Framework for Joint Estimation of Material Strength and Stress Profile}
\end{center}
  \medskip
} \fi

\bigskip
\begin{abstract}
For structural components whose failure is governed by the weakest-link theory, existing reliability models typically either assume that the underlying mechanical model is known or neglect to exploit the spatial information contained in observed failure locations. In practice, however, idealized mechanical models may systematically deviate from the actual stress due to simplifying or incorrect assumptions. To address this limitation, we propose a hierarchical Bayesian weakest-link model that jointly estimates the latent material strength and stress profile from paired failure load and failure zone observations. In our formulation, the stress profile is estimated via a B-spline basis expansion, and the non-differentiable weakest-link mechanism is approximated by a differentiable Softmin function to account for unobserved material flaws and facilitate Bayesian inference. Simulation studies demonstrate that the proposed framework provides accurate and robust estimation under various experiment configurations. Applied to a real-data analysis of Douglas-fir crossarms, the proposed model identifies systematic deviations from idealized beam theory that cannot be captured by deterministic stress derivations. 
\end{abstract}

\noindent%
\vfill

\newpage
\spacingset{1.8} 

\section{Introduction}

In structural engineering, material heterogeneity, defects, and stress fields collectively contribute to the variability in the failure behavior of structural components. To characterize this underlying uncertainty, probabilistic models of failure are essential for reliability analysis and engineering practice. 
The weakest-link theory, proposed by \citet{peirce1926weakestlink} and formalized by \citet{weibull1939statistical}, posits that a system fails at the initiation of its most severe flaw rather than upon reaching a global mean strength threshold. That theory, and the corresponding Weibull distribution, has been applied to a broad range of materials, including timber \citep{madsen1986size}, ceramics \citep{petrovic1987weibull, lamon1988statistical}, glass \citep{beason1998basis}, carbon-fiber composites \citep{el1989estimation}, and concrete \citep{bavzant1991statistical}.

Statistical weakest-link models have found wide usage due to their applicability to both uniform and non-uniform stress fields via the effective volume approach \citep{batdorf1974statistical}.  
Under a uniform stress field, where the stress remains constant, the failure probability depends on the volume of the component and the magnitude of that constant stress. Under non-uniform stress fields, the effective volume is approximated by discretizing the structure into elements of approximately uniform stress. The overall failure probability is then the complement of the product of element-wise survival probabilities. Stochastic finite element methods extend this approach by propagating uncertainty in material properties into the stress calculations \citep{gutierrez2000stochastic, graham2003analysis}. In either case, the results obtained are highly dependent on having an accurate mechanical model for stress, as weakest-link failure predictions are sensitive to errors in the calculated stress fields. 
In practice, idealized mechanical models may systematically deviate from the actual effective stress due to factors such as uncertainty in material properties, simplifying assumptions regarding structural configuration and boundary conditions, and unmodeled local effects \citep{mottershead2011sensitivity, kiran2025state}. Moreover, material and stress heterogeneity are often coupled when observing failures: a failure observed in a low-stress peripheral zone provides substantially stronger evidence of extreme material weakness than a failure occurring at the peak-stress location. To decouple these sources of heterogeneity, a modeling framework is required for simultaneously recovering the latent material strength distribution and the actual stress profile from failure load and location data. Such a model must also explicitly incorporate the impact of observed material defects on local strength variation.

Bayesian methods offer a principled framework for combining prior engineering knowledge with experimental data to estimate and quantify uncertainty in model parameters. Early applications in this area include Bayesian inference of Weibull parameters \citep{soland1969bayesian} along with extensions for censored data and risk assessment \citep{siu1998bayesian, zhang2006bayesian}, and hierarchical treatments of failure time analysis \citep{hamada2008bayesian}. More recently, Bayesian approaches have been developed for strength prediction, which include the models proposed by \citet{wong2016quantifying} and \citet{fan2023knots} for lumber with observed material defects. However, \citet{wong2016quantifying} did not exploit information on failure location, and \citet{fan2023knots} considered only the case of uniaxial tension loading with uniform stress across the specimen. 
To our knowledge, no established method has been designed to jointly estimate the material strength distribution and the latent stress profile from failure data.

To address the limitations of existing approaches, we propose a hierarchical Bayesian weakest-link model that jointly recovers the latent material strength and stress profile from paired observations of failure load and failure location. Our model contributes three key innovations in this context. First, we decouple the latent material strength from the stress profile by modeling the log-strength as a linear mixed model with material defects as covariates and a specimen-level random effect, to capture both observed and unobserved sources of material variability. 
Second, we represent the latent stress profile using non-parametric B-splines subject to a sum-to-zero constraint. This flexible formulation allows the stress profile to be learned directly from observed failure locations without imposing a restrictive parametric form. Consequently, while priors can be informed by idealized mechanical theory, posteriors are free to adapt to deviations from that theory supported by the data.
Third, we approximate the non-differentiable weakest-link mechanism using a differentiable Softmin function. This formulation allows unobserved material flaws to contribute to failure and facilitates efficient gradient-based Bayesian inference. 
We show that our construction enables all of the parameters associated with the latent material strength and stress profile to be identifiable.
A real data application to Douglas-fir crossarms demonstrates that the proposed framework recovers systematic deviations from idealized beam theory, most notably elevated stresses due to pin-hole shear, that a fixed theoretical stress profile cannot capture.

The remainder of the paper is organized as follows. Section~\ref{sec:experiments} describes the experimental configuration of the crossarm tests and formalizes the underlying mechanical principles governing failure in structural components. Section~\ref{sec:method} introduces the hierarchical Bayesian model and establishes identifiability of all model parameters. In Section~\ref{sec:simulation}, we conduct a series of simulation studies to evaluate the model's ability to recover latent stress profiles under different loading configurations and data-generating mechanisms. Section~\ref{sec:realdata} analyzes the dataset of Douglas-fir crossarms and compares our results against benchmarks based on a deterministic stress profile. Finally, Section~\ref{sec:conclusion} provides concluding remarks and discusses future research directions.

\section{Background and Motivating Example}\label{sec:experiments}

\subsection{Douglas-fir Crossarm Dataset}

The experimental data that motivates this study comes from a large-scale destructive testing program of Douglas-fir crossarms \citep{anderson2021ability}. Crossarms are critical components of the electrical distribution grid whose failure can lead to significant service interruptions and public safety risks. The test setup mimics the real-world mechanical loading of utility crossarms using an overhanging cantilever configuration. Each specimen (measuring $3.75'' \times 4.75'' \times 96''$) was pinned to a steel beam at its center and one end, with an actuator applying a load at the free end. The load was applied at an angle of 17.5$^\circ$ to the horizontal axis of the arm, inducing biaxial bending. Spacers were integrated between the pinned supports to allow for free deflection during the loading process. Further details of the experiment and data can be found in \citet{anderson2019reassessment}.

Prior to testing, each specimen is partitioned into a series of non-overlapping longitudinal zones to localize fracture initiation (Figure~\ref{fig:zones}). During the experiment, both the failure load and the zone in which fracture initiates are recorded. These paired observations, i.e., of failure load and failure location, are the primary source of information exploited in our analysis. In addition, detailed measurements of knots, including their sizes and locations, were collected for each specimen.

\subsection{Material Strength and Flaws}

The latent strength of each zone is primarily governed by the stochastic distribution of material flaws or defects. In structural timber, these manifest as natural growth characteristics, most notably knots and grain deviations \citep{kretschmann2010mechanical}. Knots, formed from branch inclusions, force the primary wood fibers to deviate from a straight longitudinal path, thereby interrupting fiber continuity and inducing localized stress concentration \citep{cramer1983model, tang1984stress}. Research indicates that the impact of these defects on material strength depends on a complex interplay between their size, shape, and the local slope of grain \citep{cramer1988exploring, kretschmann2010mechanical}.

\begin{figure}[H]
    \centering
    \begin{subfigure}[b]{0.7\textwidth}
        \centering
        \includegraphics[trim={0 50 0 50},clip,width=.9\linewidth]{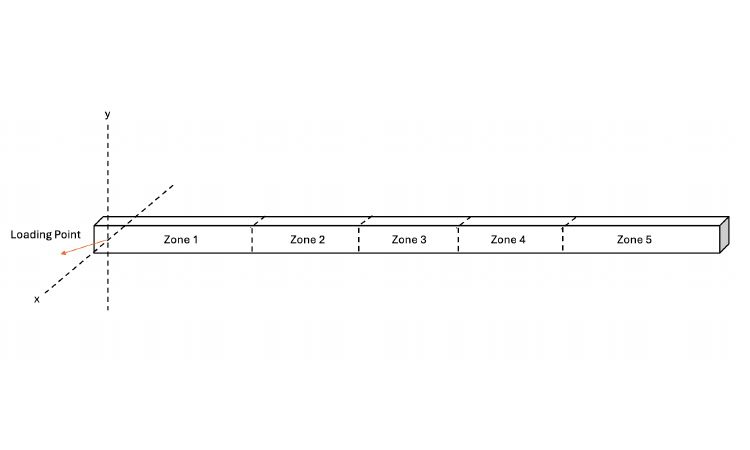}
        \caption{Discretization of zones from the overhanging cantilever configuration experiment.}
        \label{fig:zones}
    \end{subfigure}
    \begin{subfigure}[b]{0.45\textwidth}
        \centering
        \includegraphics[width=\linewidth]{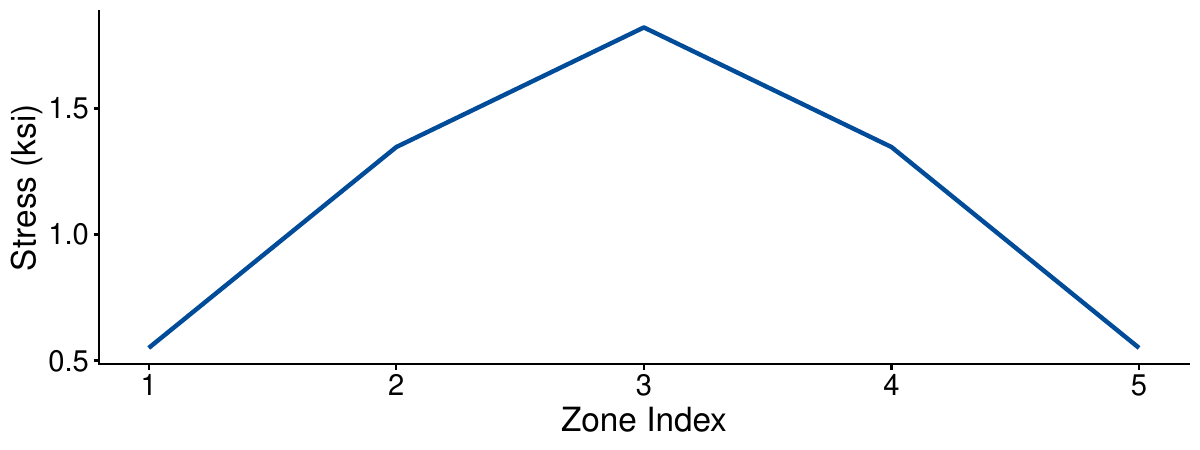}
        \caption{Theoretical tensile stress profile under the experimental loading configuration.}
        \label{fig:theory_dist}
    \end{subfigure}  
    \begin{subfigure}[b]{0.45\textwidth}
        \centering
        \includegraphics[width=\linewidth]{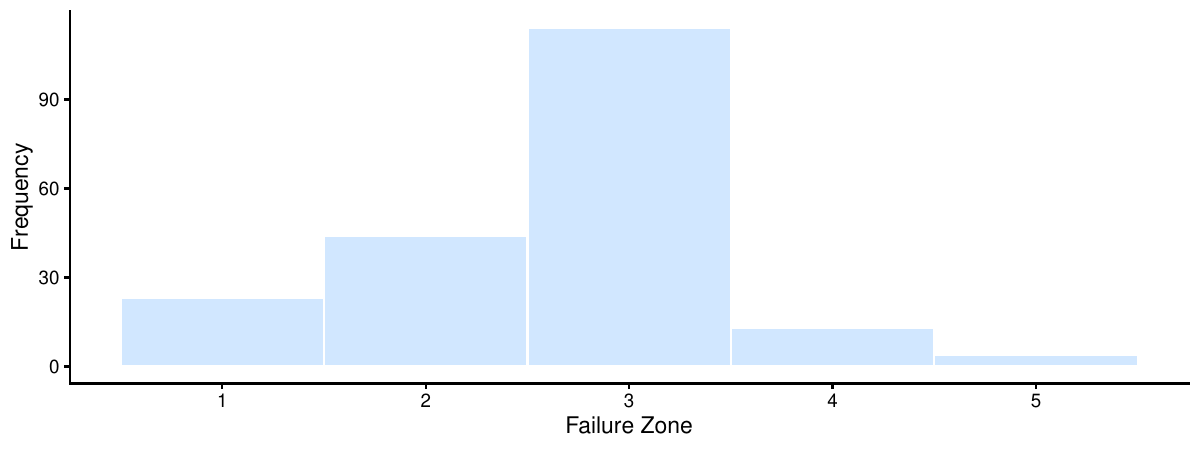}
        \caption{Distribution of failure zones observed in the experiment.}
        \label{fig:failure_zone_dist}
    \end{subfigure}  
    \caption{Motivating example of Douglas-fir crossarms under an overhanging cantilever configuration experimental setup.}
    \label{fig:motivating_example}
\end{figure}

Current industry grading standards \citep{wwpa2022, wclib2022} assess only the most prominent knot or knot cluster within a specimen, ignoring the collective spatial arrangement of defects across the beam's span \citep{wong2016quantifying}. \citet{fan2023knots} demonstrated that the spatial arrangement of knots across the zones of a specimen, rather than the characteristics of any single dominant knot, is a stronger predictor of localized failure. This directly motivates the zone-specific covariate structure adopted in Section~\ref{sec:method}, in which observed defects are incorporated separately for each longitudinal zone.

\subsection{Stress Profiles}\label{subsec:stress_profiles}

Material strength and flaws alone do not determine failure. A zone fails only when the stress induced by the external load exceeds its local strength. Consequently, accurate failure prediction requires both the probability distribution of material strength and the stress profile generated by the loading configuration. 

The simplest loading configuration is uniaxial tension, where the specimen is gripped at both ends and pulled axially, generating a uniform stress profile across the longitudinal axis. In this scenario every zone experiences identical internal stress and, in the absence of material flaws, every zone would have an equal probability of failure. However, most real-world structural components do not experience uniform stress. Bending stresses are a common example of non-uniform stress: the bending moment varies along the span, and hence the stress varies substantially along the longitudinal axis of the beam \citep{boresi2002advanced, courtney2005mechanical}. For timber crossarms subjected to an overhanging cantilever loading configuration, the bending stress reaches its maximum near the center support and decreases toward both the end support and the free end of the beam (Figure~\ref{fig:theory_dist}; see Section \ref{sec:stress_analysis} of the supplementary material for the derivation based on classical beam theory).
 
Classical beam theory provides a nominal stress profile for this loading configuration, but systematic discrepancies between the theoretical and effective stress fields commonly arise in practice. Factors such as slight misalignment of the support pins, imperfect boundary conditions, and unmodeled shear stresses cause the true internal stress field to deviate from the theoretical profile. As noted by \cite{anderson2021ability}, the physical pinning mechanism introduces localized stress concentrations around the pin holes, leading to a notable number of shear failures in Zone 1  (Figure~\ref{fig:failure_zone_dist}). Such failures are not predicted by the theoretical tensile stress profile, and the discrepancy is particularly difficult to characterize when these stress perturbations interact with material flaws. Rather than assuming the theoretical stress profile is exact, this motivates us to treat the effective stress profile as unknown and to estimate it jointly with the latent strength of each zone, given the observed failure loads and failure locations.

While our motivating example is based on structural timber, the framework we propose is applicable when each of the following holds in the experimental setup: (i) the system can be partitioned into a series of discrete longitudinal zones, (ii) failure is governed by the weakest-link theory, and (iii) the experiment records both a failure load and a failure location.

\section{Methodological Framework}\label{sec:method}

\subsection{The Strength Model}\label{sec:base_strength}

Consider a specimen partitioned into a sequence of $r$ \footnote{The number of zones $r$ is an experimental/modeling choice that balances spatial resolution against statistical precision: finer discretization provides more spatial information about the stress profile but reduces the number of observed failures per zone, potentially leading to sparse data and unstable estimates, while coarser discretization may obscure meaningful variation in the stress profile.  In practice, $r$ should be chosen to reflect the spatial scale of the expected stress variation. } discrete longitudinal zones. Let $\boldsymbol{y}_i = (y_{i1}, \dots, y_{ir})^\top \in \mathbb{R}^{r}$ represent the latent material strength vector for the $i$-th specimen, where each element $y_{ij}$ corresponds to the localized strength of the $j$-th zone. Under idealized, defect-free conditions, the material strength is assumed to be homogeneous across all zones. For each specimen $i \in [n]$ and zone $j \in [r]$, we begin by defining the baseline latent log-strength as:
\begin{equation}\label{eq:strength0}
    \log(y_{ij}^{(0)}) = \mu + \theta^{(0)}_i,
\end{equation}
where $\mu$ is the global mean 
and $\theta^{(0)}_i \sim \mathcal{N}(0, \sigma^2_{\theta^{(0)}})$ is a random effect capturing inter-specimen variability among defect-free specimens. For timber, this variability might be attributable to inherent physiological factors, such as moisture content, dry density, and growth-ring orientation.

To account for material flaws that may weaken the system, let $\mathbf{X}_{ij}$ be a vector of observed covariates (e.g., the measured characteristics of knots in timber), and $\theta^{(1)}_i$ be a random effect to account for unmeasured or latent flaws affecting the overall strength of the specimen. The defect-augmented latent strength is modeled as:
\begin{equation}\label{eq:strength1}
\log(y_{ij}) = \log(y_{ij}^{(0)}) + \mathbf{X}^\top_{ij}\boldsymbol{\beta} + \theta^{(1)}_i, \end{equation} 
where $\boldsymbol{\beta}$ is a vector of strength-reduction coefficients. 
Substituting \eqref{eq:strength0} into \eqref{eq:strength1} yields the log-linear mixed model:
\begin{equation}
\log(y_{ij}) = \mu + \mathbf{X}^\top_{ij}\boldsymbol{\beta} + \theta_i, \label{eq:full_strength_model}
\end{equation} where $\theta_i = \theta^{(0)}_i + \theta^{(1)}_i$ represents the total inter-specimen variation. In the absence of external validation data from defect-free specimens, the variance components $\sigma^2_{\theta^{(0)}}$ and $\sigma^2_{\theta^{(1)}}$ are not separately identifiable. Therefore, we define a composite random effect $\theta_i \sim \mathcal{N}(0, \sigma^2_{\theta})$, where $\sigma^2_{\theta} = \sigma^2_{\theta^{(0)}} + \sigma^2_{\theta^{(1)}}$, aggregating all sources of inter-specimen heterogeneity into a single stochastic term.

\subsection{Linking the Strength and Stress Profiles}

Since the latent stress profile is often non-uniform (recall Section \ref{subsec:stress_profiles}), 
two identical zones may exhibit different failure probabilities solely due to differences in the stress they experience. To link the strength and stress profiles, we introduce a vector of scaling factors, $\boldsymbol{\psi} = (\psi_1, \dots, \psi_r)^\top$, common to all specimens under the same loading configuration. This vector normalizes the latent strength of each zone to enable direct comparison across zones. 
The resulting adjusted strength vector for specimen $i$ is then denoted as $\boldsymbol{z}_i = (z_{i1}, \dots, z_{ir})^\top \in \mathbb{R}^{r}$, with $z_{ij} = y_{ij} / \psi_j$, where each $\psi_j$ acts as an adjustment factor for zone $j \in [r]$. In this formulation, high values of $\psi_j$ correspond to zones under high stress, where a higher realized failure load $y_{ij}$ is required to reflect the same adjusted strength. Conversely, in low-stress zones (low $\psi_j$), failure occurs only when the latent strength is low enough to be triggered by a minimal load, meaning that these zones typically survive. Log-transforming this relationship recovers the form of a log-linear additive model:
\begin{equation}\label{eq:int_logcap}
\log(z_{ij}) = \log(y_{ij}) - \log(\psi_j) = \mu + \mathbf{X}^\top_{ij}\boldsymbol{\beta} + \theta_i - \log(\psi_j),
\end{equation}
which partitions the adjusted log-strength into four fixed and random effect components.

\subsection{Representing Stress Profiles via B-Spline Basis Expansion}

Rather than imposing a parametric functional form on the zone-level stress values, we employ a B-spline basis expansion to provide a flexible representation of $\log(\boldsymbol{\psi})$ across the $r$ zones:
\begin{equation}\label{eq:basis}
    \log(\psi_j) = \sum_{g=1}^{G}\mathbf{B}_{jg}\gamma_{g},
\end{equation}
where $\mathbf{B}$ is an $r \times G$ basis matrix, with the $j$-th row consisting of the evaluations of the $G$ basis functions at the location of zone $j$. The number of basis functions is given by $G = m + d + 1$, where $m \geq 0$ is the number of internal knots and $d \geq 0$ is the spline degree, with the associated basis coefficients represented by the vector $\boldsymbol{\gamma} = [\gamma_1, \dots, \gamma_G]^\top \in \mathbf{R}^G$.
The B-spline parameterization is flexible enough to capture a wide range of stress profiles, and yields a smooth continuous curve in the limit as $r \to \infty$ to match engineering intuition.

To ensure the model remain identifiable, we impose a sum-to-zero constraint on the log-scaling factors: $\sum_{j \in [r]} \log(\psi_j) = 0$. In the absence of this constraint, an arbitrary constant could be added to $\mu$ and subtracted from $\log(\psi_j)$ without changing the likelihood. Under this constraint, $\log(\boldsymbol{\psi})$ can be interpreted as log-normalized stress, facilitating a direct and unitless comparison with the theoretical stress profile derived from mechanical theory.

\subsection{Approximating the Weakest-Link Mechanism}\label{sec:failure_zone_lik}

Let $f_i \in [r]$ denote the observed failure zone for specimen $i \in [n]$. Under a deterministic weakest-link mechanism, the specimen fails at the zone with the lowest adjusted strength: 
$f_i = \arg\min_{j \in [r]} \{ z_{ij}\}$, or equivalently, 
$f_i = \arg\min_{j \in [r]} \{ \log(z_{ij})\}$. 
However, the $\arg\min$ operator is non-differentiable, which hinders gradient-based Bayesian inference. A deterministic hard minimum also assumes that failure is determined solely by the single lowest modeled strength value, potentially failing to account for unmeasured zone-level material flaws. Although specimen-level heterogeneity is captured through $\theta_i$ in \eqref{eq:full_strength_model}, unmeasured heterogeneity at the zone level is not explicitly modeled. To address these computational and conceptual limitations, we approximate the discrete weakest-link failure process using a differentiable Softmin function. 

While the standard Softmax function is conventionally utilized in multinomial logistic regression and neural networks to identify the maximum value in a set \citep{bridle1990probabilistic}, here we adapt the Softmin function to identify the minimum latent strength. We model the failure zone $f_i$ as a Categorical random variable: 
\begin{equation}\label{eq:categorical}
    f_i \sim \text{Categorical}(\boldsymbol{w}_i),
\end{equation} 
where $\boldsymbol{w}_i = (w_{i1}, w_{i2}, \dots, w_{ir})$. The probability of failure in zone $j$ for specimen $i$ is specified via a modified Softmin function: 
\begin{equation}\label{eq:softmin}
    w_{ij} = \frac{\exp(-k \log(z_{ij}))}{\sum_{m \in [r]} \exp(-k \log(z_{im}))},
\end{equation}
where $k > 0$ controls the sharpness (or concentration) of the failure probability. Because the likelihood of $f_i$ depends on the product $k \log(z_{ij})$, any scaling of $k$ could be perfectly offset by an inverse scaling of $\log(z_{ij})$. To ensure identifiability, $k$ must be fixed as a hyperparameter.

\remark As $k \rightarrow \infty$, the model recovers the deterministic weakest-link behavior with $w_{ij} \rightarrow 1$ for $j = \arg\min_{m}  \{ \log(z_{im}) \}$ and $w_{ij} \rightarrow 0$ otherwise. Conversely, as $k \to 0$, $w_{ij} \rightarrow 1/r$, implying that the failure location is purely random and independent of the material's local adjusted strength.

\remark \label{remark_wij} Note that by expanding the $\log(z_{ij})$ term, 
  \begin{align*}
    w_{ij}
    &= \frac{
        \exp\left(-k\left(\mu + \mathbf{X}_{ij}^{\top}\boldsymbol{\beta}
        + \theta_i - \mathbf{B}_j^{\top}\boldsymbol{\gamma}\right)\right)
      }{
        \sum_{m=1}^{r}
        \exp\left(-k\left(\mu + \mathbf{X}_{im}^{\top}\boldsymbol{\beta}
        + \theta_i - \mathbf{B}_m^{\top}\boldsymbol{\gamma}\right)\right)
      } 
    = \frac{
        \exp\left(k\mathbf{B}_j^{\top}\boldsymbol{\gamma}
        - k\left(\mu + \mathbf{X}_{ij}^{\top}\boldsymbol{\beta}\right)\right)
      }{
        \sum_{m=1}^{r}
        \exp\left(k\mathbf{B}_m^{\top}\boldsymbol{\gamma}
        - k\left(\mu + \mathbf{X}_{im}^{\top}\boldsymbol{\beta}\right)\right)
      }.
  \end{align*}
  The factor $\exp(-k\theta_i)$ cancels, i.e., 
  $w_{ij}$ depends only on $(\mu, \boldsymbol{\beta}, \boldsymbol{\gamma})$
  and is independent of $\theta_i$.

\subsection{Failure Load Likelihood}\label{sec:failure_load_lik}

Together with the failure zone, our modeling framework must also account for the loading magnitude that induces failure. 
Through mechanical theory, the observed failure load can be converted into the maximum stress sustained by the material immediately prior to failure, referred to as the ultimate strength. Let $y_i^{obs}$ denote the log-ultimate strength \footnote{Depending on the loading configuration, $y_i^{obs}$ may represent the log-ultimate bending strength (i.e., the maximum bending stress before fracture), the log-ultimate tensile strength (i.e., the maximum tensile stress reached during a tensile test), or the log-ultimate compressive strength (i.e., the maximum compressive stress before crushing).} of specimen $i$. Conditional on the failure zone $f_i$, we model $y_i^{obs}$ using a normal distribution: 
$   y_i^{obs} \mid f_i, \log(\boldsymbol{y}_i) 
    \sim \mathcal{N}(\log(y_{i,f_i}), \sigma^2_{\epsilon})$,
where $\log(y_{i,f_i})$ is the latent log-strength of the failing zone as defined in \eqref{eq:strength1}, and $\sigma_{\epsilon}$ represents the standard deviation of measurement error. As in \citet{wong2016quantifying}, we treat $\sigma_{\epsilon}$ as a fixed, small number because the measurement uncertainty of the testing apparatus tends to be minimal. Computationally, a fixed $\sigma_{\epsilon}$ also ensures identifiability of the latent strength parameters and facilitates gradient-based posterior sampling.

\subsection{Prior Specification}\label{sec:prior}

We employ flat priors for the global intercept and fixed effects. Specifically, we assign $\mu \sim \mathrm{U}[a, b]$, where $a$ and $b$ represent the physical bounds of the log-strength, and $\boldsymbol{\beta} \sim \mathcal{N}(\mathbf{0}, 100^2 \mathbf{I})$. Following common practice for variance components, we assign an inverse-gamma prior to the random effect variance, i.e., $1/\sigma^2_\theta \sim \Gamma(0.001, 0.001)$ \citep{spiegelhalter1996bugs}. 

The prior specification for $\boldsymbol{\gamma}$ and $\log(\boldsymbol{\psi})$ depends on whether prior knowledge of the loading configuration is available. When no such knowledge is available, independent Gaussian priors $\boldsymbol{\gamma} \sim \mathcal{N}(\mathbf{0}, \sigma^2_\gamma \mathbf{I})$ can be assigned, treating each coefficient as a penalized regression term \citep{fong2010bayesian}. This induces a weakly informative prior on $\log(\boldsymbol{\psi}) = \mathbf{B} \boldsymbol{\gamma}$ with mean zero, corresponding to a uniform stress profile. The parameter $\sigma_\gamma$ is assigned the weakly informative prior $\sigma_\gamma \sim \mathcal{N}_+(0,1)$.  Under this specification, the stress profile is learned entirely from the failure zone likelihood.

When prior knowledge about the loading configuration is available, a physics-informed prior can be placed on $\boldsymbol{\gamma} \sim \mathcal{N}\!\left(\boldsymbol{\gamma}_{\mathrm{theory}},\,
\sigma^2_\gamma\mathbf{I}\right)$, with
\(\boldsymbol{\gamma}_{\mathrm{theory}} = (\mathbf{B}^\top\mathbf{B})^{-1}\mathbf{B}^\top\log(\boldsymbol{\psi}_{\mathrm{theory}}),\)
where $\log(\boldsymbol{\psi}_{\mathrm{theory}})$ is the log-normalized theoretical stress profile derived from mechanical theory. Specifically, the nominal stresses $\boldsymbol{\sigma}$ are evaluated at each zone according to the loading configuration, and log-normalized to satisfy the sum-to-zero constraint:
\begin{equation}\label{eq:log_psi_theory}
  \log(\psi_{\mathrm{theory},j})
  = \log(\sigma_j)
  - \frac{1}{r}\sum_{j'=1}^r \log(\sigma_{j'}),
  \quad j \in [r].
\end{equation}
The prior on $\boldsymbol{\gamma}$ implies
\begin{equation}\label{eq:physics_prior}
  \log(\boldsymbol{\psi})
  \sim \mathcal{N}\left(
  \log(\boldsymbol{\psi}_{\mathrm{theory}}),\,
  \sigma^2_\gamma\mathbf{B}\mathbf{B}^\top\right),
\end{equation}
which treats the theoretical log-normalized stress profile as a prior mean while allowing the experimental data to update the posterior. This approach provides a principled way to quantify discrepancies between theoretical and empirical stress profiles.

Finally, to enforce the identifiability requirement $\sum_{j=1}^{r} \log (\psi_j) = 0$, we utilize a soft-constraint that treats the sum as a pseudo-observation with a tight variance $\sum_{j=1}^{r} \log(\psi_j) \sim \mathcal{N}(0, (0.001r)^2)$.

\subsection{Posterior Distribution}

Let $\boldsymbol{y} = (\boldsymbol{y}_1, \dots, \boldsymbol{y}_n)^\top$
and $\boldsymbol{\theta} = (\theta_1, \dots, \theta_n)^\top$.
We collect all primary model parameters into the vector $\boldsymbol{\eta} = (\mu,\, \boldsymbol{\beta},\,
  \sigma^2_\theta,\, \boldsymbol{\gamma})$.
Since the stress profile is a fixed physical quantity,
$\sigma_\gamma$ acts as a regularization parameter rather than a primary parameter of interest. 
The remaining quantities $\log(\boldsymbol{\psi}) = \mathbf{B}\boldsymbol{\gamma}$,
$\log(\boldsymbol{y})$, and $\log(\boldsymbol{z})$ are deterministic
functions of $\boldsymbol{\eta}$ and $\boldsymbol{\theta}$ rather than free parameters. 
The joint posterior distribution of $(\boldsymbol{\eta}, \boldsymbol{\theta},
\sigma_\gamma)$ given the observed data $(\boldsymbol{y}^{obs}, \boldsymbol{f})$
factorizes as:
\begin{equation}\label{eq:posterior}
  p(\boldsymbol{\eta}, \boldsymbol{\theta}, \sigma_\gamma
  \mid \boldsymbol{y}^{obs}, \boldsymbol{f})
  \propto
  p(\boldsymbol{\eta}, \sigma_\gamma)
  p(\boldsymbol{\theta} \mid \sigma^2_\theta)
  \prod_{i=1}^{n}
  p(f_i \mid \boldsymbol{\eta}, \theta_i)
  \prod_{i=1}^{n}
  p(y_i^{obs} \mid f_i, \boldsymbol{\eta}, \theta_i),
\end{equation}
where the prior distribution for $\boldsymbol{\eta}$ and $\sigma_\gamma$ is specified in Section~\ref{sec:prior}, the random-effects distribution $p(\boldsymbol{\theta} \mid \sigma^2_\theta)$ is described in Section~\ref{sec:base_strength}, the failure zone likelihood $p(f_i \mid \boldsymbol{\eta}, \theta_i)$ is defined in Section~\ref{sec:failure_zone_lik}, and the failure load likelihood $p(y_i^{obs} \mid f_i, \boldsymbol{\eta}, \theta_i)$ is defined in Section~\ref{sec:failure_load_lik}.

\subsection{Identifiability}

We conclude our modeling framework by establishing that the parameter vector $\boldsymbol{\eta}$ is identifiable from the observed data under mild conditions.

\begin{assumption}\label{ass:design}
  Define the augmented vector $\tilde{\mathbf{X}}_{ij} = (1,\, \mathbf{X}_{ij}^{\top})^{\top} \in \mathbb{R}^{p+1}$. Then we assume the matrix $\sum_{i,j}\tilde{\mathbf{X}}_{ij}
  \tilde{\mathbf{X}}_{ij}^{\top}$ has full rank $p+1$.
\end{assumption}

\begin{assumption}\label{ass:bspline}
  The basis matrix $\mathbf{B} \in \mathbf{R}^{r \times G}$ has full column rank $G$, where $G = m + d + 1$ for $m \geq 0$ internal knots and spline degree $d \geq 0$. Note that this is satisfied whenever the $r$ zone evaluation points are distinct and $G \leq r$.
\end{assumption}

\begin{assumption}\label{ass:constraints}
  The sum-to-zero constraint $\sum_{j=1}^r \log(\psi_j) = 0$ is imposed, and the hyperparameters $k > 0$ and $\sigma_\epsilon > 0$ are fixed.
\end{assumption}

\begin{proposition}[Identifiability]\label{prop:ident}
  Under Assumptions~\ref{ass:design}--\ref{ass:constraints}, the parameter vector $\boldsymbol{\eta} = (\mu, \boldsymbol{\beta}, \sigma^2_\theta, \boldsymbol{\gamma})$ is identified:
  if $p(y^{obs}, f \mid \boldsymbol{\eta}) = p(y^{obs}, f \mid \boldsymbol{\eta}')$
  for all $(y^{obs}, f)$, then $\boldsymbol{\eta} = \boldsymbol{\eta}'$.
\end{proposition}
The proof is provided in Section \ref{sec:proofs} of the supplementary material.

\section{Simulation Studies}\label{sec:simulation}

Via a series of simulated experiments, we assess the model's ability to accurately recover structural parameters and stress profiles under various loading configurations and data-generating mechanisms.

\subsection{Data Generation Process}

Each specimen $i \in [n]$ is partitioned into $r = 5$ discrete longitudinal zones, with $\mu = 2.5$ representing the mean log-strength of a clear (defect-free) beam. To simulate natural inter-specimen variability, each piece is adjusted by the random effect $\theta_i \sim \mathcal{N}(0, \sigma^2_\theta)$. We vary the number of specimens $n \in \{100, 200, 500\}$ and inter-specimen standard deviation $\sigma_\theta \in \{0.1, 0.3\}$ to assess the quality of parameter recovery across a range of scenarios. 

Within each zone $j$, we incorporate zone-specific covariates including the Total Knot Count (TKC), representing the number of distinct knots within each zone, and the Total Knot Area (TKA), which aggregates the recorded areas of all knots in each zone. We simulate these covariates based on empirical distributions derived from our real dataset: we independently sample $\text{TKC}_{ij} \sim \text{Poisson}(2)$ and $\text{TKA}_{ij} \sim \Gamma(0.4, 0.6)$. The associated strength reduction coefficients are set to $\beta_{tka} = -0.03$ and $\beta_{tkc} = -0.03$.

To represent the stress profile, we employ a B-spline of degree $d=3$ without internal knots ($m=0$), which results in $G=4$ basis functions. The vector of log-scaling factors $\log(\boldsymbol{\psi})$ is modeled using $\eqref{eq:basis}$ according to the specific experimental scenarios detailed in Section \ref{sec:sim_experiment}.  
Given the log-strength $\log(y_{ij})$ and the structural penalty $\log(\boldsymbol{\psi})$, the adjusted log-strength $\log(z_{ij})$ is calculated via~\eqref{eq:int_logcap}. Following the weakest-link theory, the observed failure $f_i$ is determined according to~\eqref{eq:categorical} and~\eqref{eq:softmin}. We investigate the impact of the failure zone likelihood by varying the true sharpness parameter $k$ in the weakest-link mechanism. To account for the inherent variability of timber, we evaluate three scenarios: $k = 10$, where failure is highly concentrated at the theoretical weakest zone;  $k = 6$,  which reflects moderate material unpredictability;  and $k = 2$, which represents high unpredictability. Finally, the observed log-ultimate strength is recorded as the log-strength of the failing zone after the addition of a small measurement error with $\sigma_{\epsilon} = 0.01$. We simulate 100 datasets given these parameters. For the purpose of the simulation study, we assume no prior knowledge of the stress profiles when fitting the model, in order to robustly evaluate parameter identifiability.

The simulated data pairs $(f_i, y_i^{obs})$ were used to fit the proposed Bayesian model via No-U-Turn sampler (NUTS) in CmdStan \citep{cmdstanr}. To ensure efficient convergence and avoid posterior geometry-induced divergent transitions, we utilize a non-centered parameterization for the random effects $\theta_i$ with $\theta_i = \tilde\theta_{i} \sigma_\theta$, where $\tilde\theta_{i} \sim \mathcal{N}(0, 1)$.

\subsection{Loading Configurations}\label{sec:sim_experiment}

We evaluate model performance under two loading experiments, each commonly encountered in mechanical engineering applications. 

\begin{enumerate}
\item Overhanging Cantilever Configuration: This scenario utilizes a symmetric triangular stress profile that peaks at the interior pinned support. Motivated by classical beam theory, the spline coefficients are set to $\boldsymbol{\gamma} = (-0.1, 0.1, 0.1, -0.1)$ to recover a triangular-like stress gradient across the $r$ discrete zones. The choice of $\boldsymbol{\gamma}$ also ensures failures are observed in every zone for each selected $k$. Consequently, $\log(\boldsymbol{\psi})$ are generated according to \eqref{eq:basis} and subsequently centered such that $ \sum_{j}\log(\psi_j)= 0$.

\item Uniaxial Tension Configuration: In this scenario, all zones are subjected to a uniform stress profile across the length of the beam. Therefore, the spline coefficients are set to $\boldsymbol{\gamma} = (0, 0, 0, 0)$, resulting in $\log(\psi_j) = 0$ for all $j \in [r]$. This configuration serves as a baseline to demonstrate that the model fitting and recovery of the parameters is robust even in the simplest setup.

\end{enumerate}

\FloatBarrier
\subsection{Simulation Results}

\subsubsection{Overhanging Cantilever Configuration}

\begin{figure}[h]
    \centering
    \includegraphics[width=\linewidth]{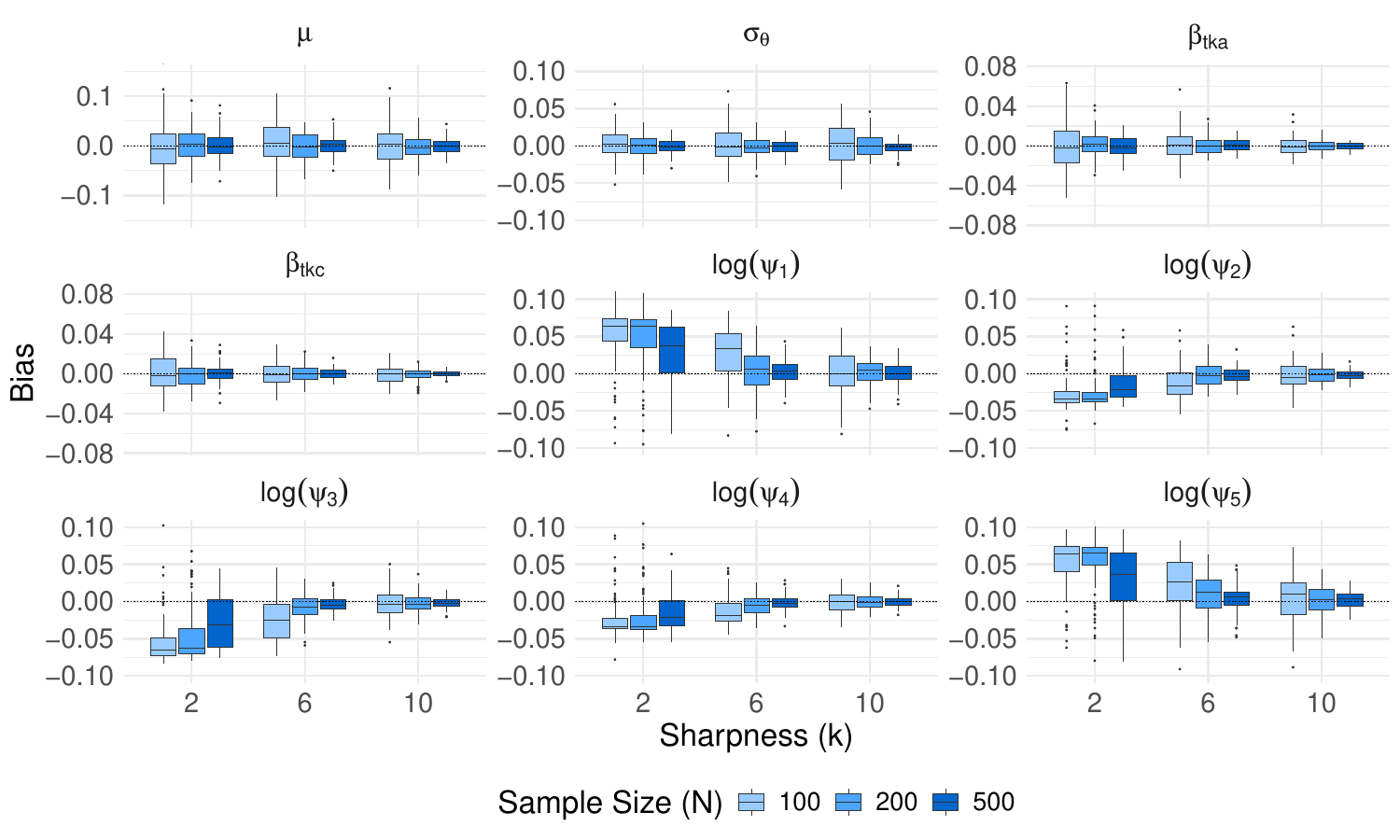}
    \caption{Bias plot of target parameters for simulation of overhanging cantilever configuration with $\sigma_\theta = 0.3$.}
    \label{fig:bias_plot_cantilever_0.3}
\end{figure}

\begin{figure}[h]
    \centering
    \includegraphics[width=\linewidth]{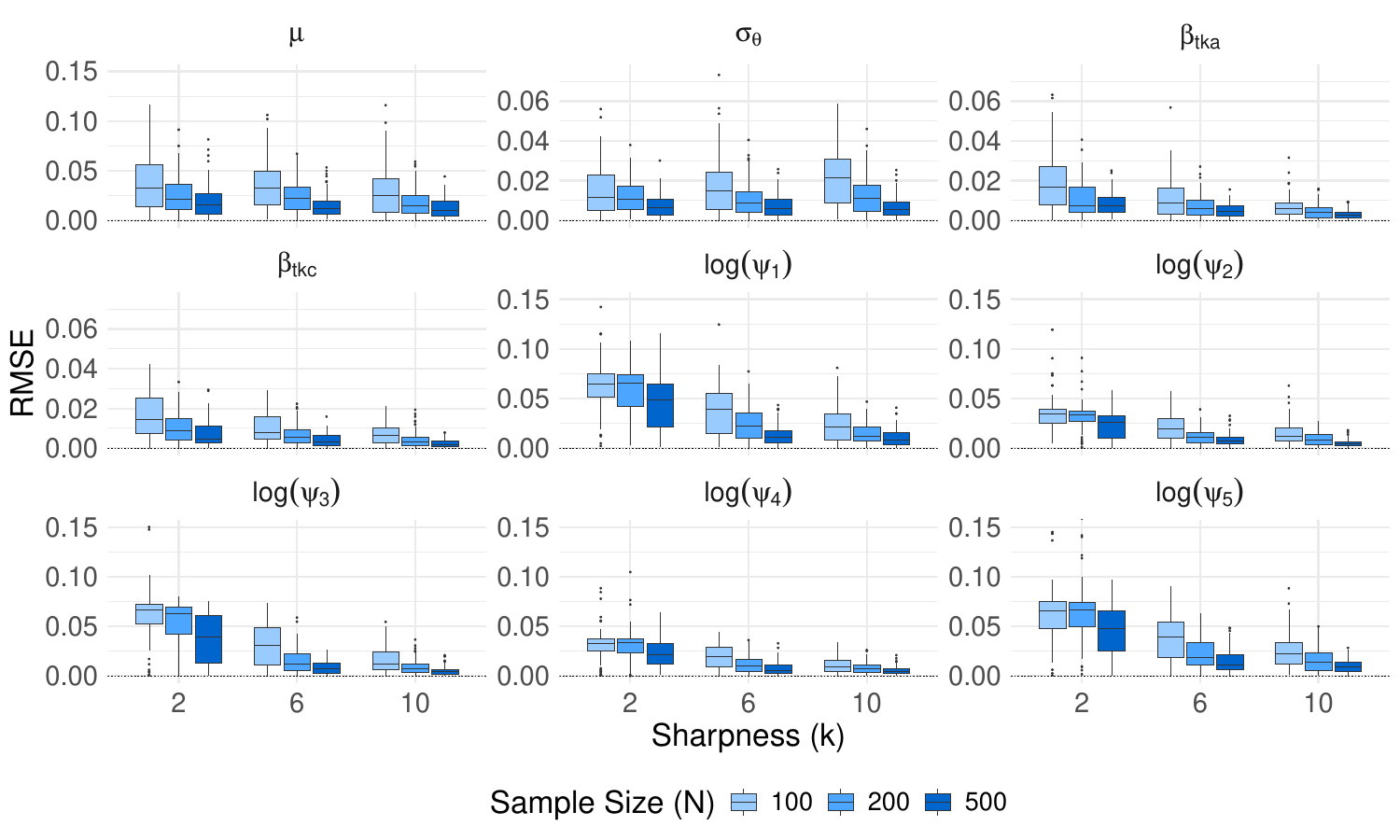}
    \caption{RMSE plot of target parameters for simulation of overhanging cantilever configuration with $\sigma_\theta = 0.3$.}
    \label{fig:rmse_plot_cantilever_0.3}
\end{figure}

Figures~\ref{fig:bias_plot_cantilever_0.3} and~\ref{fig:rmse_plot_cantilever_0.3} present the parameter bias and RMSE for the overhanging cantilever configuration under the high-noise regime ($\sigma_\theta = 0.3$). The posterior means of $\mu$, $\sigma_\theta$, $\beta_{tka}$, and $\beta_{tkc}$ are centered near zero across all configurations. 
As $n$ increases from 100 to 500, the interquartile range of the bias tightens for all parameters and the RMSE declines steadily, reflecting posterior concentration around the true values as more data become available. The effect of $k$ is most pronounced for the stress profile and knot effect parameters. While lower values of $k$ produce more diffuse Softmin weights that carry limited information about the failure zone, 
higher values of $k$ sharpen the weights and concentrate failure probability onto the weakest zone. Thus, increasing $k$ substantially reduces both the bias and RMSE of $\log(\boldsymbol{\psi})$ and narrows the bias distributions of $\beta_{tka}$ and $\beta_{tkc}$. This indicates that a sharper weakest-link mechanism better identifies the correct stress profile from the observed failure zone data.

Figures~\ref{fig:bias_plot_cantilever_0.1} and~\ref{fig:rmse_plot_cantilever_0.1} in Section \ref{sec:sim_results_supp} of the supplementary material present the corresponding results under the low-noise regime ($\sigma_\theta = 0.1$). The strength parameters $\mu$, $\sigma_\theta$, $\beta_{tka}$, and $\beta_{tkc}$ exhibit smaller bias and RMSE compared to the high-noise regime, reflecting reduced inter-specimen variability. In contrast, the stress profile parameters $\log(\boldsymbol{\psi})$ are largely unaffected by changes in $\sigma_\theta$, confirming that the stress profile is informed exclusively by the failure zone likelihood and is therefore independent of the noise level governing the strength parameters. 

\subsubsection{Uniaxial Tension Configuration}

Figures~\ref{fig:bias_plot_tension_0.3} and~\ref{fig:rmse_plot_tension_0.3} in Section \ref{sec:sim_results_supp} of the supplementary material present the bias and RMSE for the uniaxial tension configuration under the high-noise regime ($\sigma_\theta = 0.3$), while Figures~\ref{fig:bias_plot_tension_0.1} and~\ref{fig:rmse_plot_tension_0.1} in Section \ref{sec:sim_results_supp} of the supplementary material present the corresponding results under the low-noise regime ($\sigma_\theta = 0.1$). The posterior means of $\mu$, $\sigma_\theta$, $\beta_{tka}$, $\beta_{tkc}$, and $\log(\boldsymbol{\psi})$ are centered near zero across all configurations, indicating unbiasedness even under substantial between-specimen variability. As $n$ increases, the interquartile range of the bias narrows and the RMSE declines steadily for all parameters. As in the cantilever configuration, the bias and RMSE of $\beta_{tka}$, $\beta_{tkc}$, and $\log(\boldsymbol{\psi})$ narrow as $k$ increases, while $\mu$ and $\sigma_\theta$ are insensitive to $k$, reflecting that these parameters are informed exclusively by the failure load likelihood. 

Compared to the cantilever configuration, the RMSE for $\beta_{tka}$ and $\beta_{tkc}$ stabilizes more rapidly under uniaxial tension, reflecting that the uniform stress profile eliminates the confounding between zone-specific stress profiles and defect locations, facilitating more efficient identification of the covariate effects. The stress profile parameters $\log(\boldsymbol{\psi})$ similarly achieve lower bias and RMSE at smaller sample sizes under uniaxial tension than under the cantilever configuration, consistent with a flat stress profile being easier to recover than a peaked one. As in the cantilever configuration, $\log(\boldsymbol{\psi})$ is largely unaffected by changes in $\sigma_\theta$ under both noise regimes.

Overall the simulation results demonstrate that the proposed framework recovers all parameters accurately across the combination of loading configurations, noise regimes, and values of $k$, establishing the robustness of the model to diverse data-generating mechanisms. The uniaxial tension configuration shows that the framework correctly recovers a flat stress profile as a special case, confirming that the B-spline parameterization avoids overfitting under the simplest possible loading configuration.

\FloatBarrier
\section{Application}\label{sec:realdata}

We apply our methodological framework to the dataset of Douglas-fir crossarms to evaluate its ability to decouple the latent material strength from the underlying stress profile in a real-world engineering application.

\subsection{Data description}\label{sec:data}

The experimental setup is described in
Section~\ref{sec:experiments}, where the dataset consists of 198 Douglas-fir crossarms tested under the overhanging cantilever loading configuration. While \citet{anderson2021ability} considered the crossarms as having seven overlapping regions (including top/bottom bifurcations), for our analysis we re-indexed the regions into five consecutive, non-overlapping longitudinal zones along the beam's length (Figure~\ref{fig:zones}).

To quantify material defects, we utilize both TKC and TKA as zone-specific covariates. Because bending induces both tensile and compressive stress regions, defects may exhibit stress-dependent effects on failure behavior. Given the biaxial cantilever loading configuration, two facets of the beam are primarily in tension while the opposing two are in compression. We therefore stratify each defect covariate into tension and compression components to evaluate their potentially distinct impacts. The ultimate strength in our application is measured by the modulus of rupture (MOR).

We set $G=5$ basis functions. Unlike the simulation studies, which employ a flat prior on $\boldsymbol{\gamma}$ to confirm likelihood identifiability, the real data analysis incorporates engineering knowledge through a physics-informed prior  \eqref{eq:physics_prior}. Specifically, the nominal bending stresses $\boldsymbol{\sigma}$ are evaluated at the midpoint of each zone using the bending moment diagram corresponding to the overhanging cantilever loading configuration (Section~\ref{sec:stress_analysis}), with the hyperprior $\sigma_\gamma \sim \mathcal{N}_+(0, 0.2)$.

\subsection{Model Selection}\label{sec:model_selection}
 
Since $k$ must be fixed as a hyperparameter, we employ 10-fold cross-validation to select the optimal value across $k \in [1, 10]$. For each value of $k$, we run four parallel MCMC chains, each consisting of 10,000 iterations with 5,000 discarded as warmup iterations. The dataset is partitioned into 10 folds, with each fold serving once as the test set while the remaining folds are used for training. All metrics reported in this section are out-of-sample, computed on held-out specimens not used for fitting. We evaluate performance using two complementary sets of criteria. For failure load prediction, we report RMSE and MAE between the posterior predictive mean and the observed $\log(\text{MOR})$ on the test set. For failure zone prediction, we report the Brier score of the predicted zone probabilities \citep{gneiting2007strictly}, the exact hit rate, and the top-2 accuracy. The Brier score for a categorical outcome with $r$ zones is defined as $\frac{1}{n}\sum_i\sum_{j=1}^r (\hat{w}_{ij} - \mathbf{1}[f_i = j])^2$, where $\hat{w}_{ij}$ is the posterior mean predicted probability of failure in zone $j$ for specimen $i$. The Brier score measures the mean squared difference between the predicted posterior probabilities and the observed outcomes, thereby assessing the accuracy and calibration of the entire predictive distribution rather than only the most probable class.  Since $k$ directly controls the sharpness of the Softmin failure zone probabilities, the failure zone metrics are the primary selection criterion; failure load metrics serve as a secondary check.

Table~\ref{tab:cv} summarizes the out-of-sample predictive performance across $k \in [1, 10]$. Failure load RMSE and MAE vary by at most 0.004 across all values of $k$, indicating that failure load prediction is largely insensitive to the sharpness parameter. In contrast, the failure zone metrics reveal a clearer pattern. The Brier score improves steadily from $k = 1$ to $k = 5$, where it reaches its minimum and stabilizes, before deteriorating for $k \geq 6$. This suggests that moderate values of $k$ yield better-calibrated failure zone probability estimates than either very diffuse or overly sharp Softmin weights.
The exact hit rate remains relatively stable at $k = 2$ to $5$ before declining for larger values, whereas the top-2 accuracy is highest at $k = 1$ to $2$ and decreases gradually as $k$ increases. Overall, no single value of $k$ is optimal across all evaluation metrics. We therefore select $k = 5$, as it achieves the highest hit rate and lowest Brier score, while maintaining competitive top-2 accuracy, RMSE, and MAE relative to the other values of $k$. To further assess the robustness of the inference, a comprehensive sensitivity analysis over $k \in [1, 10]$ is presented in Section~\ref{sec:sensitivity}.

\begin{table}[ht]
\caption{Comparison of out-of-sample predictive performance across $k \in [1, 10]$ on the Douglas-fir crossarm dataset.}
\label{tab:cv}
\footnotesize
\centering
\begin{tabular}{lcccccccccc}
\hline
$k$    & 1 & 2 & 3 & 4 & 5 & 6 & 7 & 8 & 9 & 10 \\\hline
RMSE & 0.2535 & 0.2543 & 0.2534 & 0.2536 & 0.2542 & 0.2538 & 0.2545 & 0.2554 & 0.2558 & 0.2577 \\ 
MAE & 0.1964 & 0.1959 & 0.1955 & 0.1963 & 0.1968 & 0.1962 & 0.1976 & 0.1984 & 0.1992 & 0.2005 \\ 
Brier score & 0.6198 & 0.6036 & 0.5980 & 0.5982 & 0.5970 & 0.5977 & 0.5996 & 0.6012 & 0.6020 & 0.6032 \\ 
Hit rate (\%) & 57.58 & 58.08 & 58.08 & 58.08 & 58.08 & 57.07 & 57.07 & 56.5657 & 56.57 & 56.06 \\ 
Top-2 (\%) & 79.29 & 79.29 & 78.79 & 78.79 & 78.28 & 77.78 & 77.27 & 77.27 & 76.77 & 76.77 \\ 
\hline
\end{tabular}
\end{table}

\FloatBarrier
\subsection{Posterior Results}\label{sec:posterior_results}
 
Having selected $k=5$ via cross-validation, we refit the model on the full dataset of 198 specimens. Posterior summaries are presented in Table~\ref{tab:data_result}. The posterior mean of $\mu = 2.3248$ represents the baseline log-strength (ksi) for the population of defect-free timber specimens. The posterior estimate of $\sigma_\theta$ indicates that there is notable between-specimen variability. Specifically, this residual variability might be attributed to timber characteristics and flaws such as grain deviation or resin pockets that are invisible to standard visual grading but significantly reduce the overall strength. 

Regarding the available covariates, total knot area on the tension faces shows a clear negative impact on strength, confirming that larger total knot areas in tension significantly reduce the structural strength of the crossarms. Conversely, the coefficient for total knot area on the compression faces is smaller with much weaker evidence of marginal effect on strength. The total knot counts for both the tension and compression faces exhibit credible intervals that overlap zero, suggesting they provide no additional explanatory power for the variation in strength.
\begin{table}[ht]
\caption{Posterior summaries of parameters of the Douglas-fir crossarm dataset.}\label{tab:data_result}
\centering
\small
\begin{tabular}{lccccc}
  \hline
Parameter & Mean & SD & Median & 2.5\% & 97.5\% \\ 
 \hline
 $\mu$ & 2.3248 & 0.0236 & 2.3256 & 2.2783 & 2.3703 \\ 
  $\sigma_{\theta}$ & 0.2598 & 0.0130 & 0.2592 & 0.2362 & 0.2872 \\ 
  $\beta_{tka}$(tension) & -0.0351 & 0.0112 & -0.0351 & -0.0569 & -0.0132 \\ 
  $\beta_{tka}$(compression) & -0.0145 & 0.0108 & -0.0146 & -0.0357 & 0.0067 \\ 
  $\beta_{tkc}$(tension) & -0.0012 & 0.0133 & -0.0013 & -0.0271 & 0.0252 \\ 
  $\beta_{tkc}$(compression) & -0.0116 & 0.0145 & -0.0119 & -0.0396 & 0.0172 \\ 
  $\gamma_1$ & -0.0466 & 0.0455 & -0.0464 & -0.1366 & 0.0423 \\ 
  $\gamma_2$ & -0.1011 & 0.0760 & -0.1012 & -0.2505 & 0.0464 \\ 
  $\gamma_3$ & 0.9564 & 0.1027 & 0.9555 & 0.7612 & 1.1629 \\ 
  $\gamma_4$ & -0.3863 & 0.1176 & -0.3842 & -0.6235 & -0.1609 \\ 
   \hline
\end{tabular}
\end{table}

\begin{figure}[H]
    \centering
    \includegraphics[width=.6\linewidth]{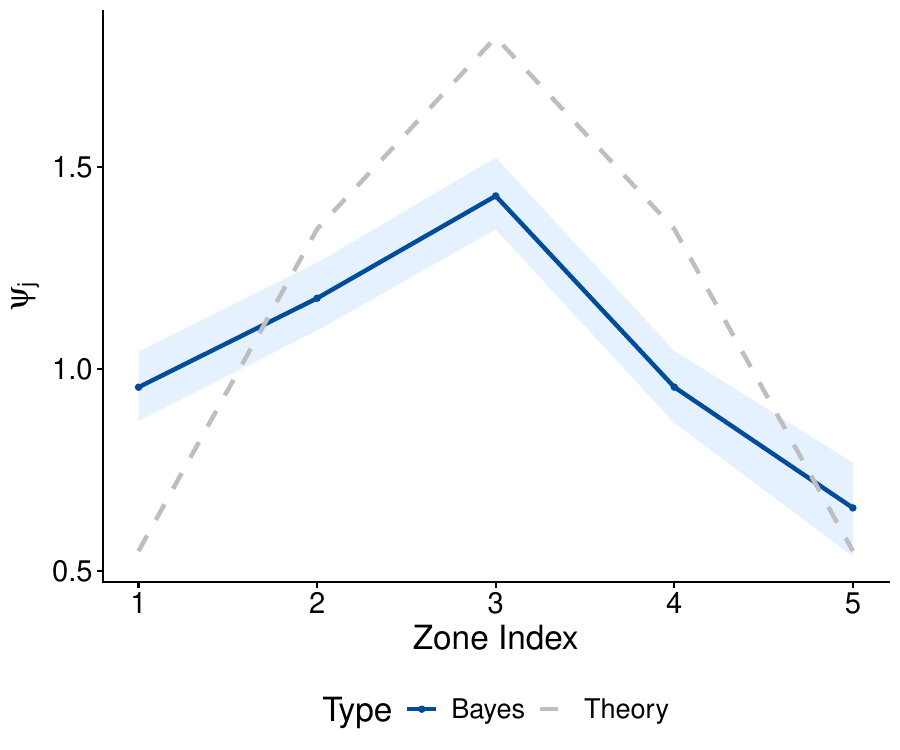}
    \caption{Comparison of the posterior distribution of $\boldsymbol{\psi}$ estimated from the Douglas-fir crossarm data with the theoretical stress profile. The blue line and shaded ribbon represent the posterior means and 95\% credible intervals from the Bayesian model, while the gray line represents the theoretical stress profile.}
    \label{fig:posterior_psi}
\end{figure}

Figure~\ref{fig:posterior_psi} presents the posterior mean and 95\% credible intervals for $\boldsymbol{\psi}$, contrasted against the theoretical stress profile $\boldsymbol{\psi}_{\mathrm{theory}}$. The posterior recovers the primary stress peak in Zone 3, consistent with the theoretical profile. However, our model reveals two systematic deviations from idealized beam theory. First, the inferred stress profile in Zone 1 is substantially higher than the theoretical value, consistent with experimental observations of shear concentrations around the pin-hole connections \citep{anderson2021ability}. Second, Zone 3 experiences a lower effective stress than predicted theoretically, likely because the elevated Zone 1 failures reduce the proportion of surviving specimens that can fail at the theoretical peak. These findings demonstrate that idealized beam theory alone cannot characterize the stress profiles of timber crossarms under real-world loading configurations.

\subsection{Comparison with Theoretical Stress Profile}

\begin{figure}[h]
    \centering
    \begin{subfigure}[b]{0.4\textwidth}
        \centering
        \includegraphics[width=.8\linewidth]{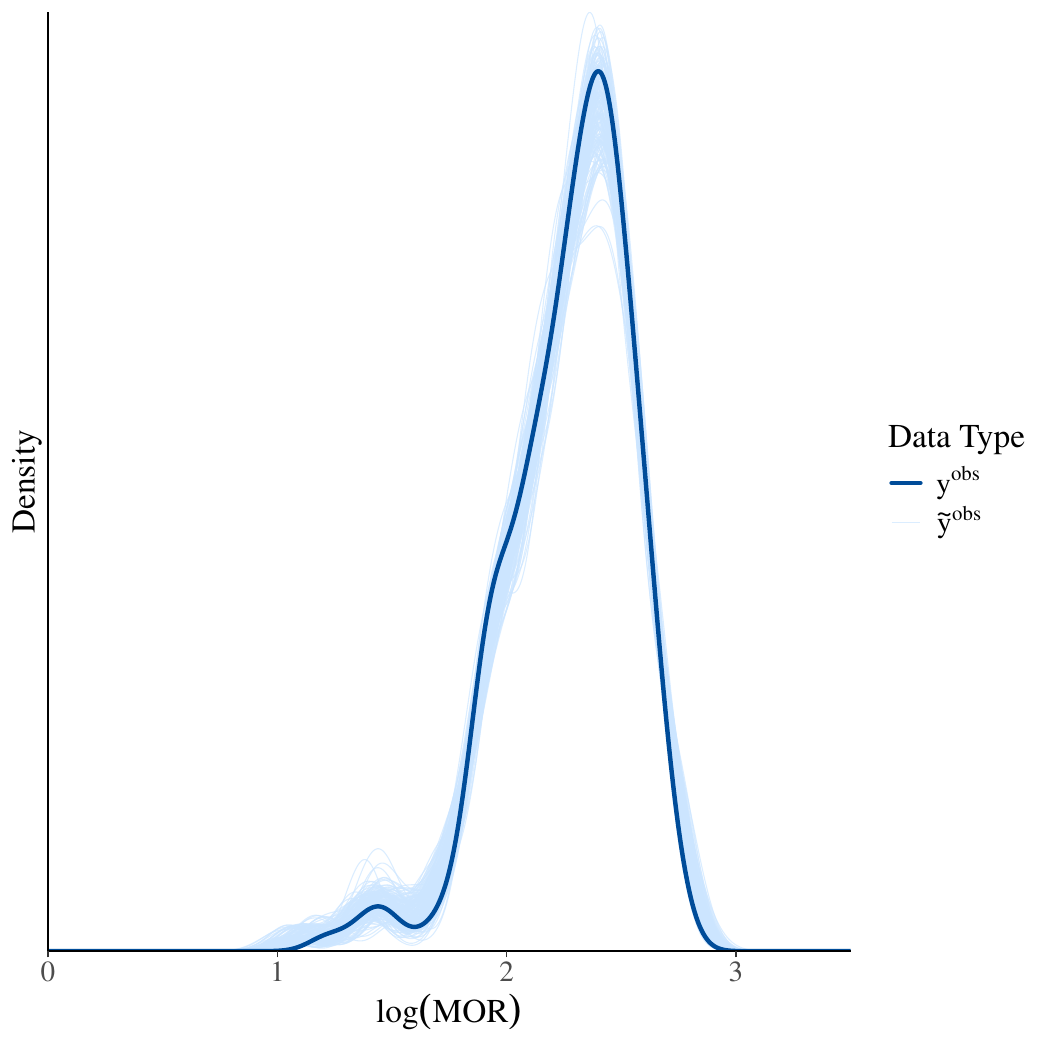}
        \caption{Comparison of $\tilde{y}^{obs}$ (physics-informed splines) and observed $y^{obs}$.}
        \label{fig:diagnostics_y}
    \end{subfigure}
    \begin{subfigure}[b]{0.4\textwidth}
        \centering
        \includegraphics[width=.8\linewidth]{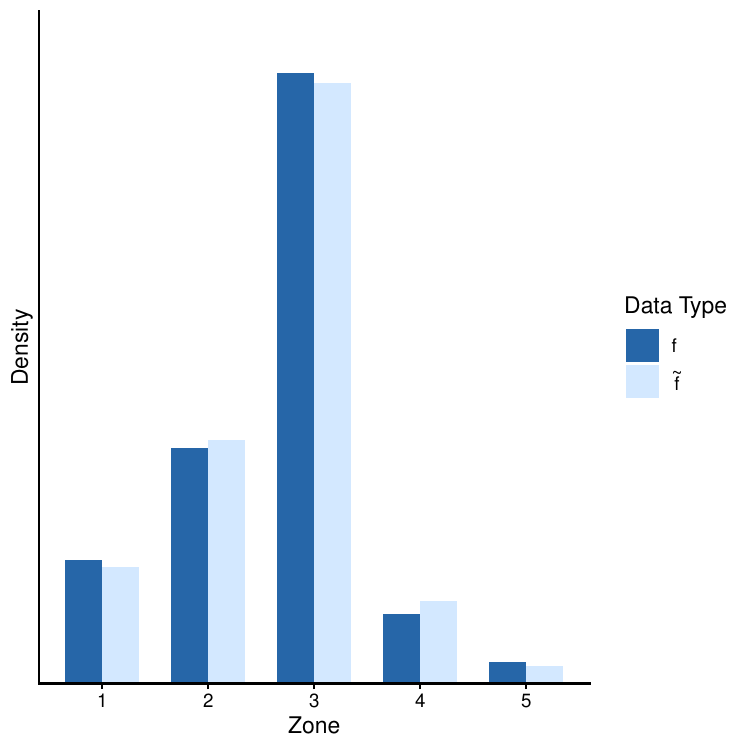}
        \caption{Comparison of $\tilde{f}$ (physics-informed splines) and observed $f$.}
        \label{fig:diagnostics_f}
    \end{subfigure}
    \begin{subfigure}[c]{0.4\textwidth}
        \centering
        \includegraphics[width=.8\linewidth]{fig/real_data_dens_plot.pdf}
        \caption{Comparison of $\tilde{y}^{obs}$ (deterministic theory) and observed $y^{obs}$.}
        \label{fig:diagnostics_y_theory}
    \end{subfigure}
    \begin{subfigure}[d]{0.4\textwidth}
        \centering
        \includegraphics[width=.8\linewidth]{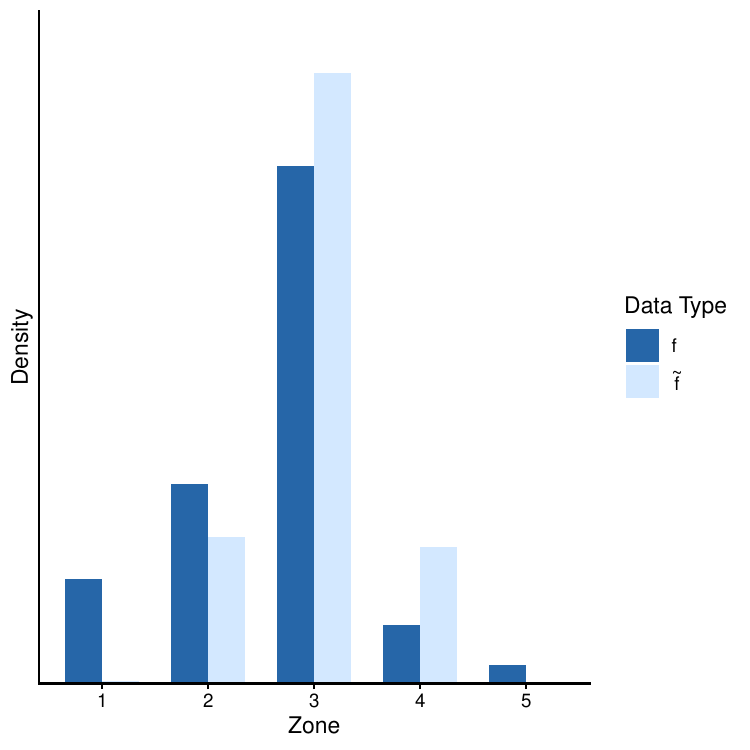}
        \caption{Comparison of $\tilde{f}$ using (deterministic theory) and observed $f$.}
        \label{fig:diagnostics_f_theory}
    \end{subfigure}
    \caption{Comparing physics-informed splines with deterministic stress profile. $\tilde{y}^{obs}$ and $\tilde{f}$ represent posterior predictive log(MOR) and failure zone respectively, while $y^{obs}$ and $f$ represent observed log(MOR) and failure zone respectively. }
    \label{fig:model_diagnostics}
\end{figure}

To demonstrate the practical utility of estimating the stress profile from data, we compare the proposed model with flexible B-splines against a baseline that fixes $\log(\boldsymbol{\psi})$ at the theoretical values $\log(\boldsymbol{\psi}_{\mathrm{theory}})$ -- equivalently, the limiting case of the physics-informed prior as $\sigma_\psi \to 0$. Both models otherwise use the same strength specification, covariates, and $k$. The adequacy of the models are  evaluated using posterior predictive checks \citep{gelman2013bayesian}.

Figure~\ref{fig:model_diagnostics} presents the posterior predictive checks for both models. For failure load prediction (Figures~\ref{fig:diagnostics_y} and~\ref{fig:diagnostics_y_theory}), both models closely reproduce the observed $\log(\text{MOR})$ distribution, including its bimodality. This confirms that the goodness-of-fit for the failure load is driven by the strength model and is insensitive to the stress profile specification. However, the failure zone diagnostics reveal important differences between the two approaches (Figures~\ref{fig:diagnostics_f} and~\ref{fig:diagnostics_f_theory}). The proposed spline model captures the observed failure zone distribution more closely, including the elevated failure rate in Zone 1 attributable to pin-hole shear concentrations \citep{anderson2021ability}. In contrast, the theoretical stress model predicts virtually no failures in the end zones, reflecting the inability of idealized beam theory to capture stresses that deviate from the triangular theoretical profile.

Table~\ref{tab:model_diagnostics} quantifies this difference using the population $\chi^2$ discrepancy \citep{agresti2013categorical}, Kullback--Leibler (KL) divergence \citep{kullback1951information}, Brier score, hit rate, and top-2 accuracy between the posterior predictive and observed failure zone distributions. The spline model outperforms the theoretical baseline on all five metrics. At the population level, the spline model achieves substantially lower $\chi^2$ ($0.004$ vs $5.102$) and KL divergence ($0.002$ vs $0.402$). These results demonstrate that the spline model substantially better replicates the aggregate empirical failure zone distribution, including the elevated Zone 1 failure rate that the theoretical baseline systematically misses. At the specimen level, the spline model also outperforms on all three per-specimen metrics: it achieves a lower Brier score, confirming better calibration of the predicted zone probability distributions, as well as higher exact hit rates and top-2 accuracy, demonstrating its ability to more reliably identify the most vulnerable zone and the two most vulnerable zones for each specimen. Together these results demonstrate that estimating $\log(\boldsymbol{\psi})$ from failure location data substantially improves the model's ability to characterize structural vulnerability at both the population and specimen level.

\begin{table}[htbp]
\caption{Posterior predictive discrepancy between
predicted and observed failure zone distributions. Physics-informed splines model is the proposed model, while the Deterministic theory model is a baseline that fixes $\log(\boldsymbol{\psi})$ at the theoretical values $\log(\boldsymbol{\psi}_{\mathrm{theory}})$.}
\label{tab:model_diagnostics}
\centering
\small
\begin{tabular}{lccccc}\hline
Model & Population $\chi^2$ & KL divergence & Brier score & Hit rate (\%) &  Top-2 (\%)\\\hline
Physics-informed splines & 0.0035 & 0.0018 & 0.5830 & 58.0808 & 78.2828 \\
Deterministic theory & 5.1024 & 0.4016 & 0.6138 & 57.5758 & 75.7576 \\\hline
\end{tabular}

\end{table}

\subsection{Sensitivity to Failure Zone Sharpness}\label{sec:sensitivity}

Figure~\ref{fig:sensitivity_param} shows the posterior estimates for parameters $\mu, \sigma_\theta, \beta_{tka}$ and $\beta_{tkc}$ from the sensitivity analysis over a range of $k \in [1, 10]$. The stability of the posterior estimates demonstrates the robustness of the model's conclusions. 
The most notable result is the stability of $\beta_{tka}$(tension). Across all values of $k$, the posterior mean is uniformly negative and the 95\% credible intervals exclude zero, providing strong evidence that total knot area for tension-side defects is a primary driver of structural failure. 
In contrast, the evidence for compression-side defects is weaker, with the 95\% credible intervals for $\beta_{tka}$(compression) overlapping zero for most values of $k$. Likewise, the coefficients associated with total knot count all have 95\% credible intervals that include zero, confirming their negligible marginal contribution after accounting for total knot area. The 95\% credible intervals for all four covariates shrink as $k$ increases, consistent with the findings from the simulation study under the high noise regime. While some minor deviations are seen in $\mu$ and $\sigma_{\theta}$ as $k$ increases, the core physical inferences remain unchanged. 

Figure~\ref{fig:sensitivity_psi} shows the posterior stress profile $\boldsymbol{\psi}$ across $k \in [1, 10]$. The shape of $\log(\boldsymbol{\psi})$ is highly stable: Zone 3 is identified as the primary stress peak and the asymmetric cantilever pattern is well-recovered across all values of $k$. We observe that the absolute magnitude of $\psi$ scales inversely with $k$. This reflects the fact that, as the Softmin function approaches a strict minimum, a smaller difference in the latent stress profile is sufficient to determine the failure.

\begin{figure}[h]
    \centering
    \includegraphics[width=.9\linewidth]{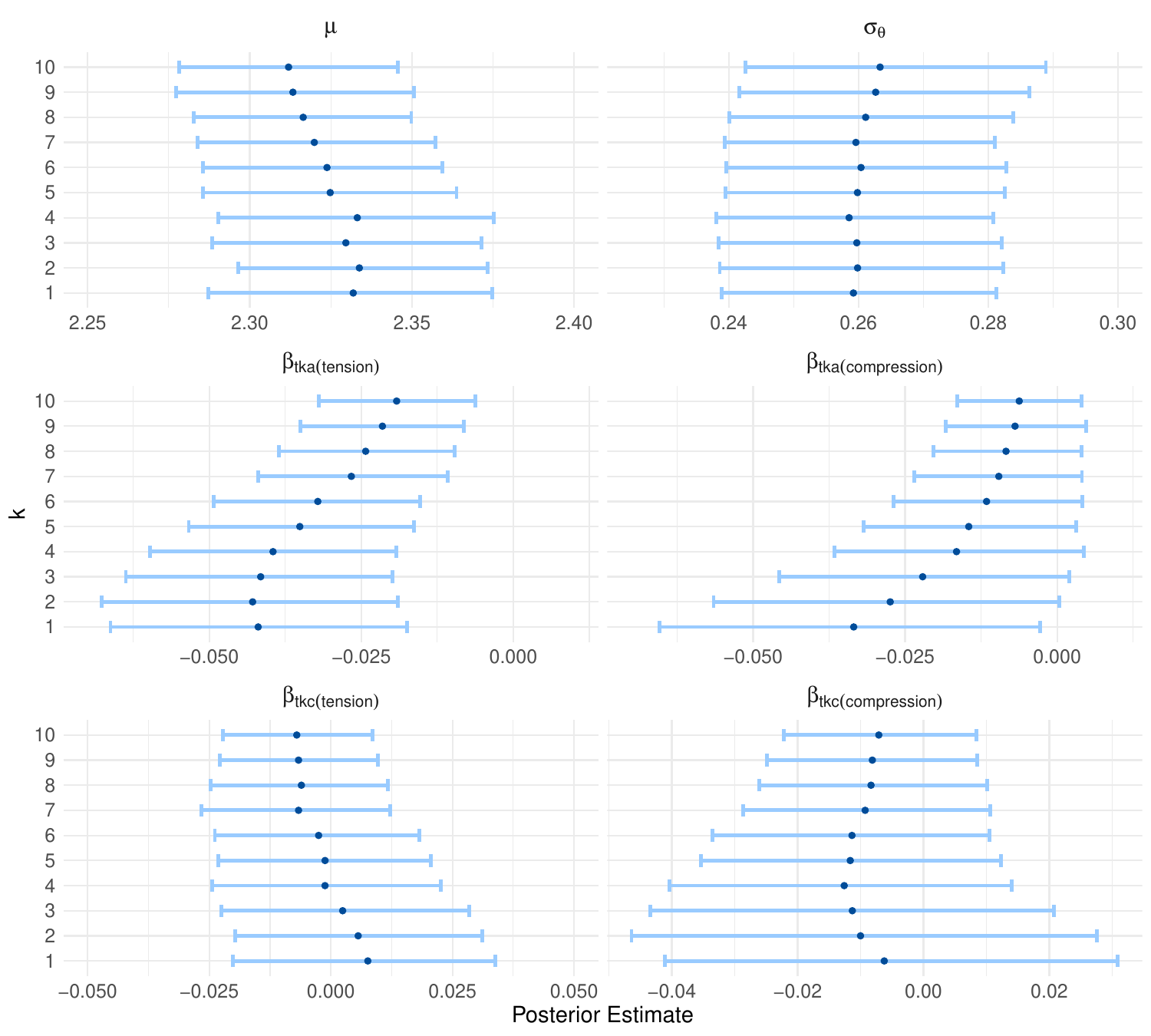}
    \caption{Sensitivity analysis for parameters $\mu, \sigma_\theta, \beta_{tka}$ and $\beta_{tkc}$ on the Douglas-fir crossarm dataset.}
    \label{fig:sensitivity_param}
\end{figure}

\begin{figure}[h]
    \centering
    \includegraphics[width=\linewidth]{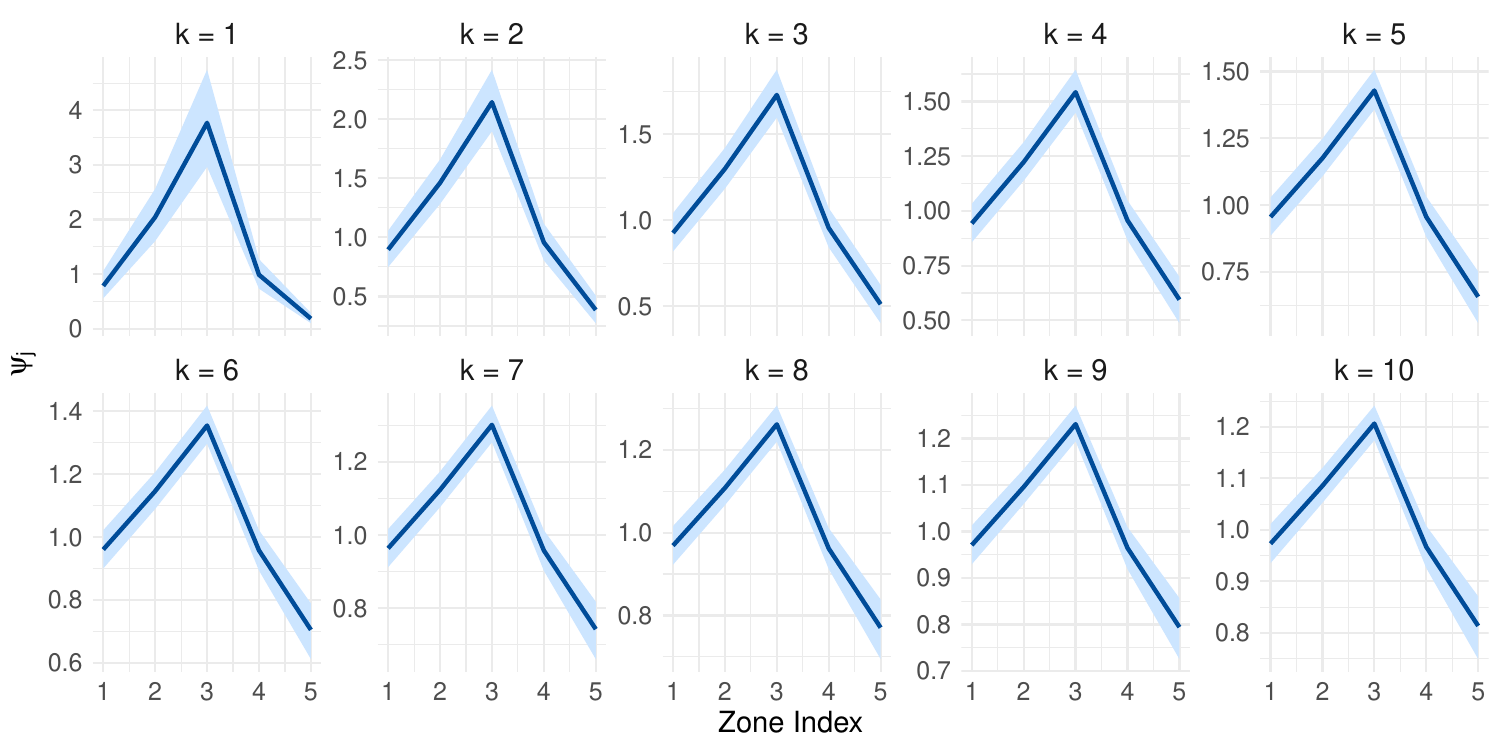}
    \caption{Sensitivity analysis for $\boldsymbol{\psi}$ on the Douglas-fir crossarm dataset.}
    \label{fig:sensitivity_psi}
\end{figure}

\section{Discussion}\label{sec:conclusion}

In this work, we developed a Bayesian weakest-link framework from paired failure load and failure location data. A key contribution of the model is its explicit decoupling of intrinsic material strength from the latent stress profile: by parameterizing the stress profile through a non-parametric B-spline basis expansion and approximating the weakest-link mechanism with a differentiable Softmin function, the model recovers the inherent stress profile informed by failure data. Simulation studies across two canonical loading configurations -- overhanging cantilever and uniaxial tension -- confirm the framework's robustness across diverse mechanical setups. 

Application to Douglas-fir crossarms under an overhanging cantilever loading reveals three notable findings. First, total knot area on the tension faces has a clear negative effect on strength, while defects on the compression faces show weaker marginal effects, and the total knot count has a negligible marginal effect. Second, the recovered stress profile reveals systematic deviations from idealized beam theory: the inferred stress profile in Zone 1 is substantially higher than the theoretical prediction, consistent with localized shear concentrations around the pin-hole connections that simple beam theory cannot capture. Third, the proposed framework significantly outperforms a baseline model that fixes the stress profile at its theoretical value. This demonstrates that treating the stress profile as an estimated latent process yields more comprehensive inference than assuming it is known from theory. Sensitivity analysis demonstrates that parameter estimates are robust across a range of sharpness values. 

Several limitations of the current framework suggest directions for future work. First, the stress scaling factor $\psi_j$ represents a single effective stress profile per zone, implicitly assuming that failure within each zone is governed by a dominant stress state. Under biaxial loading, such as the one used in the Douglas-fir crossarm experiments, the tension, compression, and shear faces experience different stress gradients. Therefore, extending the model to accommodate multi-dimensional stress profile and multiple failure modes, including simultaneous tension, compression, and shear, would allow for a more comprehensive characterization of competing failure mechanisms. 
Second, the model currently relies on manually recorded zone-level defect measurements. Recent advances in high-resolution imaging and deep learning offer a promising avenue for automating defect detection and characterizing complex grain-knot interactions that are difficult to parameterize manually \citep{fathi2020prediction, wang2024machine, feng2024machine}.  Incorporating these image-derived features into the proposed framework would provide a richer covariate structure, potentially reducing uncertainty in strength prediction and improving parameter estimation. Third, the current model captures inter-specimen heterogeneity through a single Gaussian random effect $\theta_i$. In specimens with spatially correlated defects, such as knot clusters spanning multiple zones, a zone-specific random effect structure may better capture within-specimen heterogeneity. 

\section*{Disclosure statement}\label{disclosure-statement}

The authors report there are no competing interests to declare.

\section*{Data Availability Statement}\label{data-availability-statement}

The Douglas fir crossarm dataset is publicly available from \url{https://ir.library.oregonstate.edu/concern/graduate_thesis_or_dissertations/zc77sw48x}.

\phantomsection\label{supplementary-material}
\bigskip

\section*{Supplementary Material}\label{supplementary-material}

\begin{description}
\item[Supplementary document:] Includes Supplement sections A--C referenced in Sections 2, 3, and 4 of the main text. (.pdf file)
\end{description}

\FloatBarrier
\bibliography{ref}  

\newpage
\setcounter{section}{0} 
\renewcommand\thesection{\Alph{section}}
\renewcommand\thesubsection{\thesection.\arabic{subsection}}
\renewcommand{\thefigure}{S\arabic{figure}}
\renewcommand{\thetable}{S\arabic{table}}

\setcounter{figure}{0} 

\section*{Supplementary Material for A Bayesian Weakest-Link Framework for Joint Estimation of Material Strength and Stress Profile}

\section{Theoretical Stress Profile for Crossarm Specimens}\label{sec:stress_analysis}

Because the load was applied at a $17.5^{\circ}$ angle, bending was induced about both principal axes of the cross-section ($x$ and $y$ in Figure~\ref{fig:zones}). Following the principle of superposition for linear elastic structures, the combined bending stress $\sigma$ at any longitudinal location $l$ is expressed as 
\citep{courtney2005mechanical, boresi2002advanced}:

\begin{equation}
\sigma(l) = \frac{M_x(l) \cdot y}{I_x} + \frac{M_y(l) \cdot x}{I_y}, \label{eq:bending}
\end{equation}
where
\begin{itemize}
    \item $l$ is the longitudinal distance from the loading point.
    \item $x$ and $y$ are the distances from the neutral axes to the specific point on the cross-section (typically the furthest fibers at $b/2$ and $h/2$).
    \item $M_x(l)$ and $M_y(l)$ are the bending moments about the $x$ (width-wise) and $y$ (height-wise) axes, respectively.
    \item $I_x$ and $I_y$ are the corresponding moments of inertia.
\end{itemize}

For a beam with total length $L=96''$, where the maximum moment occurs at the midpoint $l=48''$, the bending stress for each longitudinal position $l$ is defined by the following function:

\begin{equation}
    \sigma(l) = 
    \begin{cases} 
    \frac{6 P_y l}{b h^2} + \frac{6 P_x l}{h b^2} & \text{ if } l \leq 48'', \\
    \\
    \frac{6 P_y (96 - l)}{b h^2} + \frac{6 P_x (96 - l)}{h b^2} & \text{ if } l > 48''.
    \end{cases}
\end{equation}
In the above equations, $P_y$ and $P_x$ denote the components of the applied load $P$ resolved along the height ($y$) and width ($x$) axes. These components are related to the specimen's MOR (denoted as $\sigma_{max}$) as follows:
\begin{align}
    P_y = P \sin\theta &= \frac{\sigma_{max} b^2 h^2 \sin\theta}{6L (b \sin\theta + h \cos\theta)}, \\
    P_x = P \cos\theta &= \frac{\sigma_{max} b^2 h^2 \cos\theta}{6L (b \sin\theta + h \cos\theta)},
\end{align}
where the parameters are $L = 96''$, $\theta = 17.5^{\circ}$, $b=3.5''$, and $h=4.5''$.

\section{Proof for Model Identifiability}\label{sec:proofs}

\begin{proof}
  The argument proceeds in two parts.

  \bigskip
  (i) Identification of $(\mu, \boldsymbol{\beta}, \sigma^2_\theta)$. Write
  $m_{ij} = \mu + \mathbf{X}_{ij}^{\top}\boldsymbol{\beta}$. From Remark~\ref{remark_wij}, $w_{ij}$ does not depend on $\theta_i$, so the joint likelihood for specimen $i$ at failure zone $f_i = j$ is:
  \begin{align*}
    p\left(y^{obs}, f_i = j \mid \boldsymbol{\eta}\right)
    &= w_{ij} \cdot \int
      \phi\left(y^{obs};\, m_{ij} + \theta_i,\, \sigma^2_\epsilon\right)
      \cdot \phi\left(\theta_i;\, 0,\, \sigma^2_\theta\right)
      \,d\theta_i.
  \end{align*}
  The integral is a convolution of two Gaussian densities. Write $y^{obs} = m_{if_i} + \theta_i + \varepsilon_i$ where $\theta_i \sim \mathcal{N}(0, \sigma^2_\theta)$ and
  $\varepsilon_i \sim \mathcal{N}(0, \sigma^2_\epsilon)$ are independent. Since the sum of independent Gaussians is Gaussian,
  $\theta_i + \varepsilon_i
    \sim \mathcal{N}\left(0,\;
    \sigma^2_\theta + \sigma^2_\epsilon\right)$, therefore $y^{obs} = m_{ij} + (\theta_i + \varepsilon_i)
  \sim \mathcal{N}(m_{ij},\, \sigma^2_\theta + \sigma^2_\epsilon)$,
  which gives:
  \[
    p(y^{obs})
    = \phi\left(y^{obs};\, m_{ij},\,
      \sigma^2_\theta + \sigma^2_\epsilon\right).
  \]
  Hence:
  \[
    \int
      \phi\left(y^{obs};\, m_{ij} + \theta_i,\, \sigma^2_\epsilon\right)
      \cdot \phi\left(\theta_i;\, 0,\, \sigma^2_\theta\right)
    \,d\theta_i
    = \phi\left(y^{obs};\, m_{ij},\,
      \sigma^2_\theta + \sigma^2_\epsilon\right).
  \]
 Therefore the joint likelihood simplifies to:
  \begin{equation}\label{eq:joint_closed}
    p\left(y^{obs}, f_i = j \mid \boldsymbol{\eta}\right)
    = w_{ij} \cdot
      \phi\left(y^{obs};\, m_{ij},\, \sigma^2_\theta + \sigma^2_\epsilon\right).
  \end{equation}
Suppose $p(y^{obs}, f_i \mid \boldsymbol{\eta})
  = p(y^{obs}, f_i \mid \boldsymbol{\eta^{\prime}})$ for all $(y^{obs}, f_i)$.
  Fixing $f_i = j$ and using~\eqref{eq:joint_closed}:
  \[
    w_{ij}\cdot\phi\left(y^{obs};\, m_{ij},\,
    \sigma^2_\theta + \sigma^2_\epsilon\right)
    = w_{if_i}\cdot\phi\left(y^{obs};\, m^{\prime}_{ij},\,
    \sigma^{2\prime}_\theta + \sigma^{2\prime}_\epsilon\right)
    \quad \forall\, y^{obs}.
  \]
  Since $w_{ij} > 0$ for all $k > 0$
  (Assumption~\ref{ass:constraints}), divide both sides by $w_{ij}$:
  \[
    \phi\left(y^{obs};\, m_{ij},\,
    \sigma^2_\theta + \sigma^2_\epsilon\right)
    = \phi\left(y^{obs};\, m^{\prime}_{ij},\,
    \sigma^{2\prime}_\theta + \sigma^{2\prime}_\epsilon\right)
    \quad \forall\, y^{obs}.
  \]
  
  A Gaussian density is uniquely determined by its mean and variance: two Gaussian densities are equal for all argument values if and only if both their means and variances are equal. Hence,
  \[
  \mu + \mathbf{X}_{ij}^{\top}\boldsymbol{\beta}
  = \mu' + \mathbf{X}_{ij}^{\top}\boldsymbol{\beta}'
  \quad \forall\, i \in [n],\, j \in [r].
\]
Define the augmented vector
$\tilde{\mathbf{X}}_{ij} = (1,\,
\mathbf{X}_{ij}^{\top})^{\top} \in \mathbb{R}^{p+1}$
and $\tilde{\boldsymbol{\theta}} = (\tilde\mu,\,
\tilde{\boldsymbol{\beta}}^{\top})^{\top}
\in \mathbb{R}^{p+1}$, so that the mean equality
becomes:
\[
  \tilde{\mathbf{X}}_{ij}^{\top}
  \tilde{\boldsymbol{\theta}} = 0
  \quad \forall\, i \in [n],\, j \in [r].
\]
Multiplying both sides by
$\tilde{\mathbf{X}}_{ij}$ and summing over
all $i$ and $j$:
\[
  \sum_{i=1}^{n}\sum_{j=1}^{r}
  \tilde{\mathbf{X}}_{ij}
  \tilde{\mathbf{X}}_{ij}^{\top}
  \tilde{\boldsymbol{\theta}}
  =
  \left(\sum_{i=1}^{n}\sum_{j=1}^{r}
  \tilde{\mathbf{X}}_{ij}
  \tilde{\mathbf{X}}_{ij}^{\top}\right)
  \tilde{\boldsymbol{\theta}}
  = \mathbf{0}.
\]
Full column rank of
$\sum_{i,j}\tilde{\mathbf{X}}_{ij}
  \tilde{\mathbf{X}}_{ij}^{\top}$
(Assumption~\ref{ass:design}) implies
$\boldsymbol{\beta} = \boldsymbol{\beta}'$ and $\mu = \mu'$. Hence $\mu$ and $\boldsymbol{\beta}$ are identified.

From the equal Gaussian density, we also get
\[
    \sigma^2_\theta + \sigma^2_\epsilon
    = \sigma^{2\prime}_\theta + \sigma^{2\prime}_\epsilon. 
  \]
Since $\sigma^2_\epsilon$ is fixed (Assumption~\ref{ass:constraints}),  $\sigma^2_\epsilon = \sigma^{2\prime}_\epsilon$. Hence $\sigma^2_\theta = \sigma^{2\prime}_\theta$, and $\sigma^2_\theta$ is identified.

\bigskip
(ii) Identification of $\boldsymbol{\gamma}$.
  Let $a_{ij} = k(\mu + \mathbf{X}_{ij}^{\top}\boldsymbol{\beta})$ be the
  part independent of $\boldsymbol{\gamma}$. Since $\mu$ and $\boldsymbol{\beta}$ are identified from (i) and $k$ is fixed (Assumption~\ref{ass:constraints}), $a_{ij}$ is fixed. Then
  \[
    w_{ij}(\boldsymbol{\gamma})
    = \frac{\exp\left(k\mathbf{B}_j^{\top}\boldsymbol{\gamma} - a_{ij}\right)}
           {\sum_{m=1}^{r}\exp\left(k\mathbf{B}_m^{\top}\boldsymbol{\gamma}
           - a_{im}\right)}.
  \]
  Suppose $w_{ij}(\boldsymbol{\gamma}) = w_{ij}(\boldsymbol{\gamma}')$ for
  all $j \in [r]$. Then for every $j$,
  \[
    \frac{\exp\left(k\mathbf{B}_j^{\top}\boldsymbol{\gamma}\right)}
         {\sum_{m}\exp\left(k\mathbf{B}_m^{\top}\boldsymbol{\gamma}
         - a_{im}\right)}
    =
    \frac{\exp\left(k\mathbf{B}_j^{\top}\boldsymbol{\gamma}'\right)}
         {\sum_{m}\exp\left(k\mathbf{B}_m^{\top}\boldsymbol{\gamma}'
         - a_{im}\right)}.
  \]
  Let $S = \sum_{m}\exp(k\mathbf{B}_m^{\top}\boldsymbol{\gamma} - a_{im})$
  and $S' = \sum_{m}\exp(k\mathbf{B}_m^{\top}\boldsymbol{\gamma}' - a_{im})$.
  Cross-multiplying gives, for all $j$:
  \[
    \exp\left(k\mathbf{B}_j^{\top}\boldsymbol{\gamma}\right)\cdot S'
    = \exp\left(k\mathbf{B}_j^{\top}\boldsymbol{\gamma}'\right)\cdot S,
  \]
  which implies:
  \[
    \exp\left(k\mathbf{B}_j^{\top}(\boldsymbol{\gamma} - \boldsymbol{\gamma}')\right)
    = \frac{S}{S'} =: C \quad \text{for all } j \in [r],
  \]
  where $C > 0$ is a constant independent of $j$. Taking logarithms, $k\mathbf{B}_j^{\top}(\boldsymbol{\gamma} - \boldsymbol{\gamma}')
    = \log C \text{ for all } j \in [r]$,
  which means the vector $\mathbf{B}(\boldsymbol{\gamma} - \boldsymbol{\gamma}')
  \in \mathbb{R}^r$ is a constant vector:
  \[
    \mathbf{B}(\boldsymbol{\gamma} - \boldsymbol{\gamma}')
    = \frac{\log C}{k}\,\mathbf{1}.
  \]

  Summing over all zones $j \in [r]$ and applying the sum-to-zero
  constraint $\sum_{j=1}^{r}\log\psi_j = 0$
  (Assumption~\ref{ass:constraints}):
  \[
    \sum_{j=1}^{r}\left(\log\psi_j - \log\psi_j'\right)
    = \sum_{j=1}^{r}\frac{\log C}{k}
    = \frac{r\log C}{k}
    = 0,
  \]
  which forces $\log C = 0$, i.e., $C = 1$. Therefore:
  \[
    \mathbf{B}(\boldsymbol{\gamma} - \boldsymbol{\gamma}') = \mathbf{0}.
  \]
  Since $\mathbf{B}$ has full column rank (Assumption~\ref{ass:bspline}),
  this implies $\boldsymbol{\gamma} = \boldsymbol{\gamma}'$.
  
  \bigskip
  Together parts (i)--(ii) establish $\boldsymbol{\eta} = \boldsymbol{\eta}'$.
\end{proof}

\section{Simulation Results}\label{sec:sim_results_supp}

\begin{figure}[h]
    \centering
    \includegraphics[width=\linewidth]{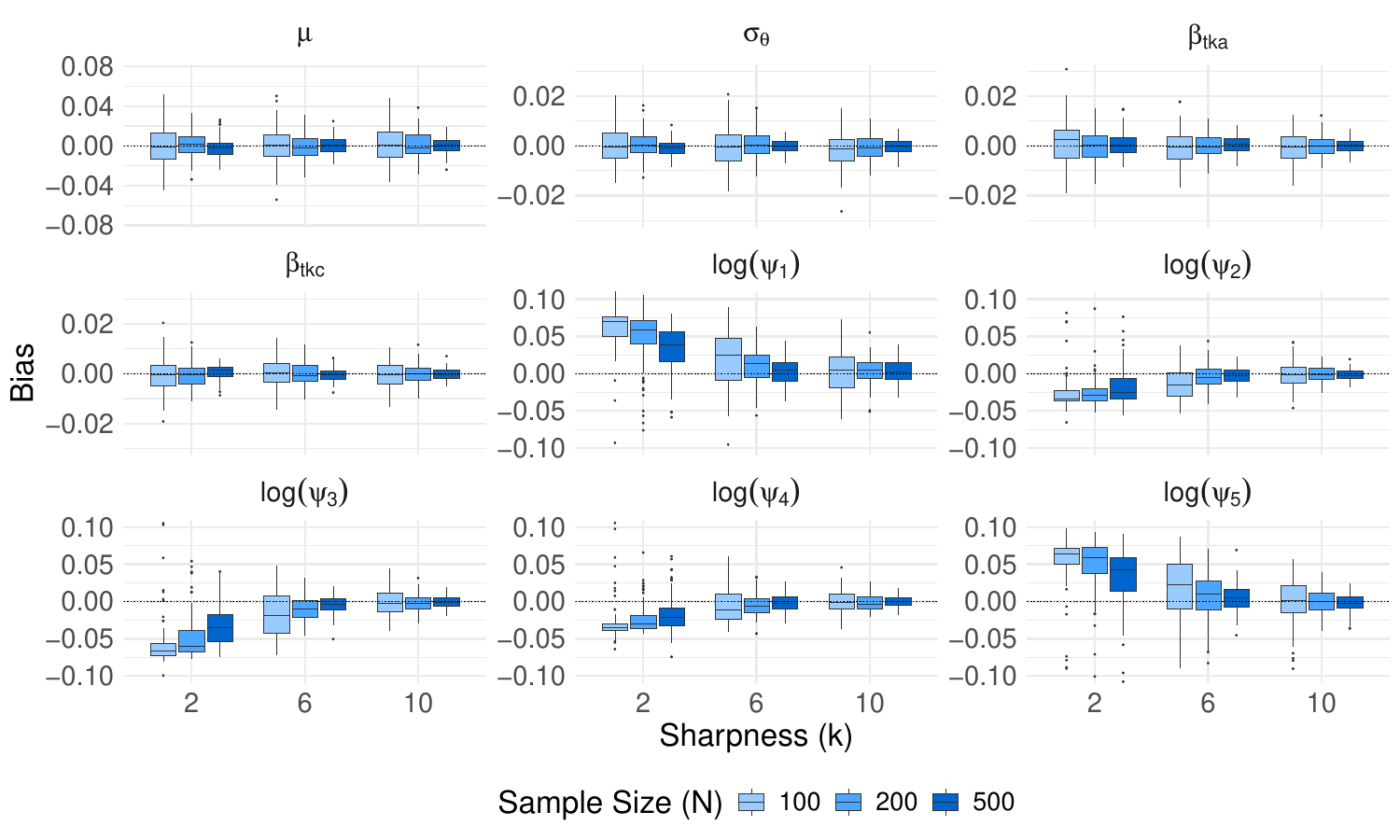}
    \caption{Bias plot of target parameters for specimens in cantilever configuration with $\sigma_\theta = 0.1$.}
    \label{fig:bias_plot_cantilever_0.1}
\end{figure}

\begin{figure}[h]
    \centering
    \includegraphics[width=\linewidth]{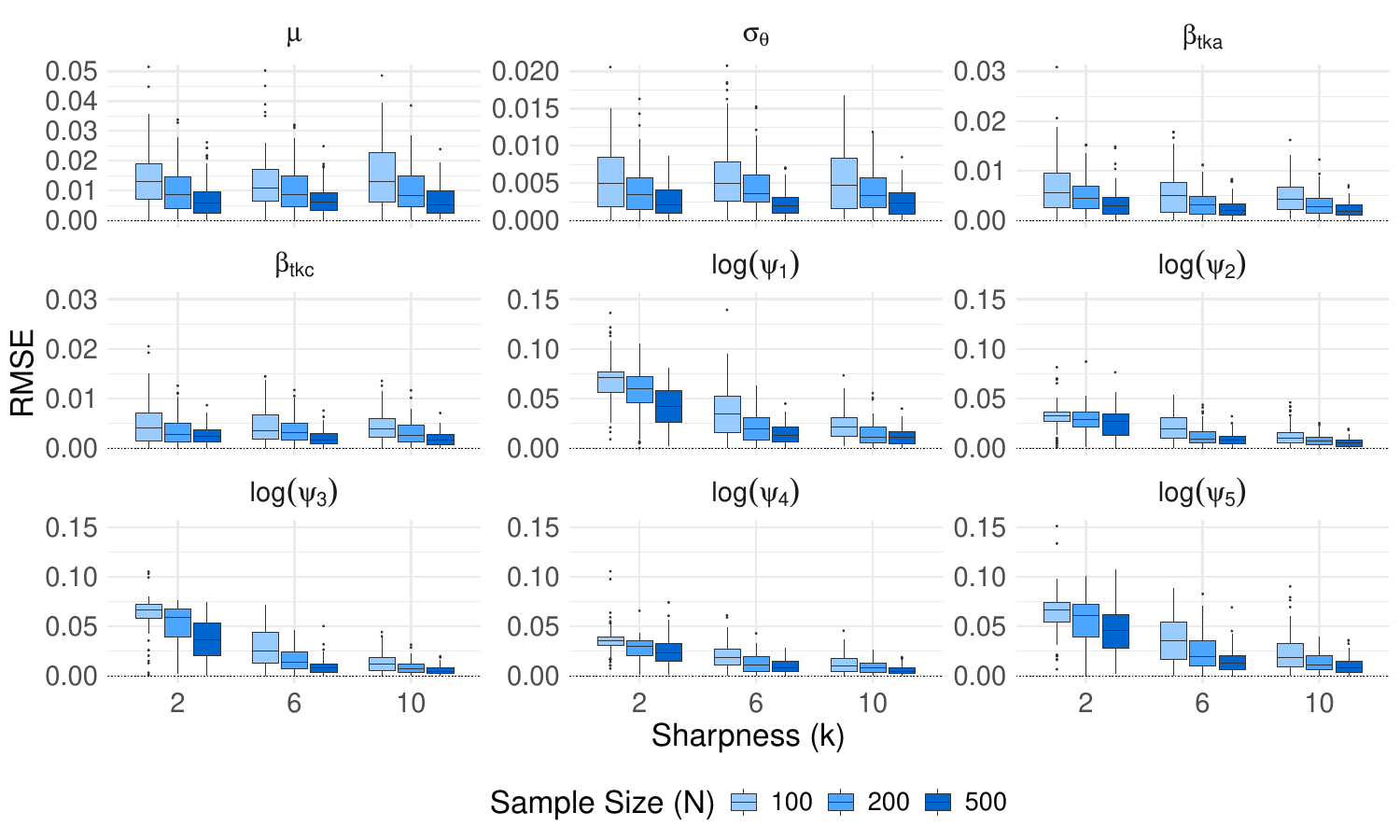}
    \caption{RMSE plot of target parameters for specimens in cantilever configuration with $\sigma_\theta = 0.1$.}
    \label{fig:rmse_plot_cantilever_0.1}
\end{figure}

\begin{figure}[h]
    \centering
    \includegraphics[width=\linewidth]{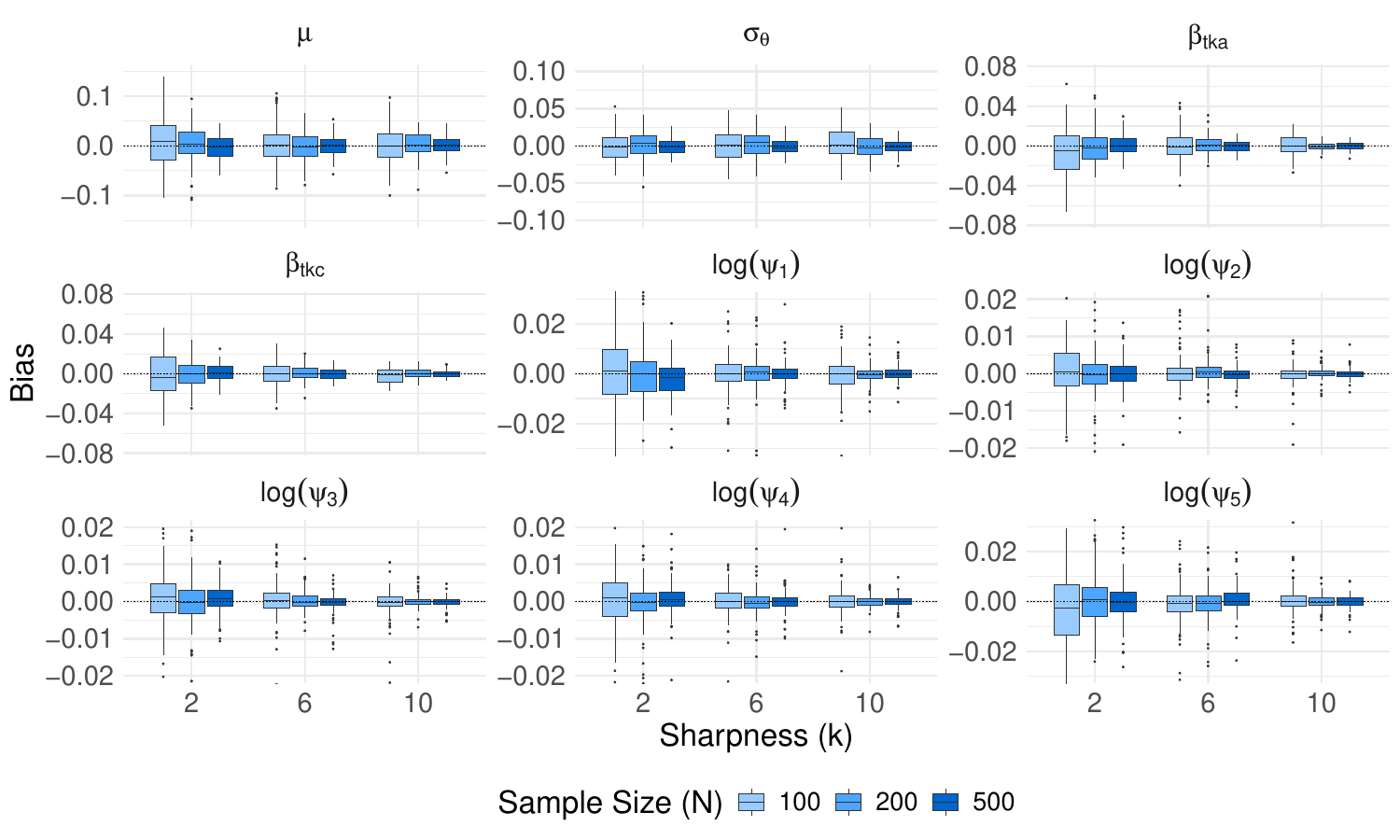}
    \caption{Bias plot of target parameters for simulation of uniaxial tension configuration with $\sigma_\theta = 0.3$.}
    \label{fig:bias_plot_tension_0.3}
\end{figure}

\begin{figure}[h]
    \centering
    \includegraphics[width=\linewidth]{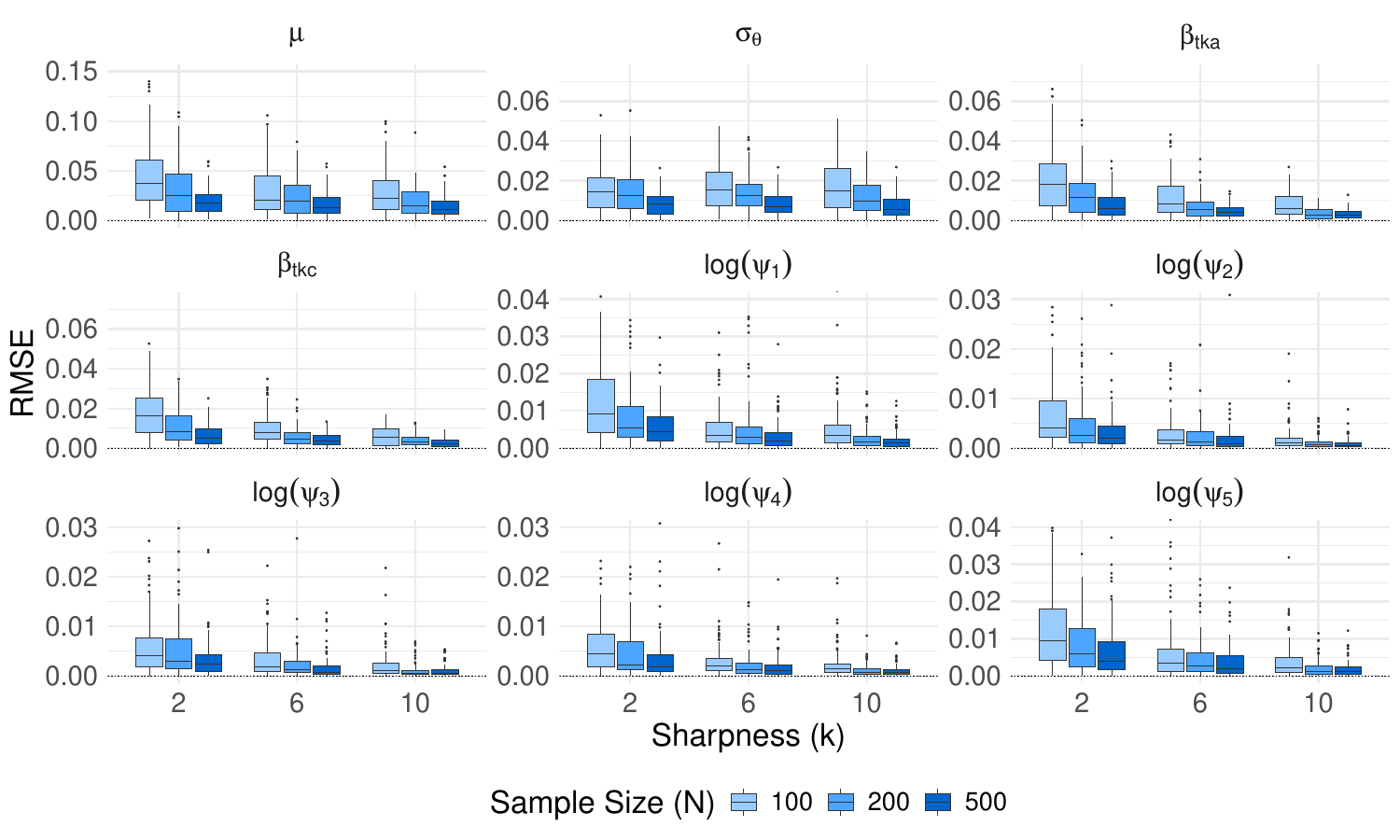}
    \caption{RMSE plot of target parameters for simulation of uniaxial tension configuration with $\sigma_\theta = 0.3$.}
    \label{fig:rmse_plot_tension_0.3}
\end{figure}

\begin{figure}[h]
    \centering
    \includegraphics[width=\linewidth]{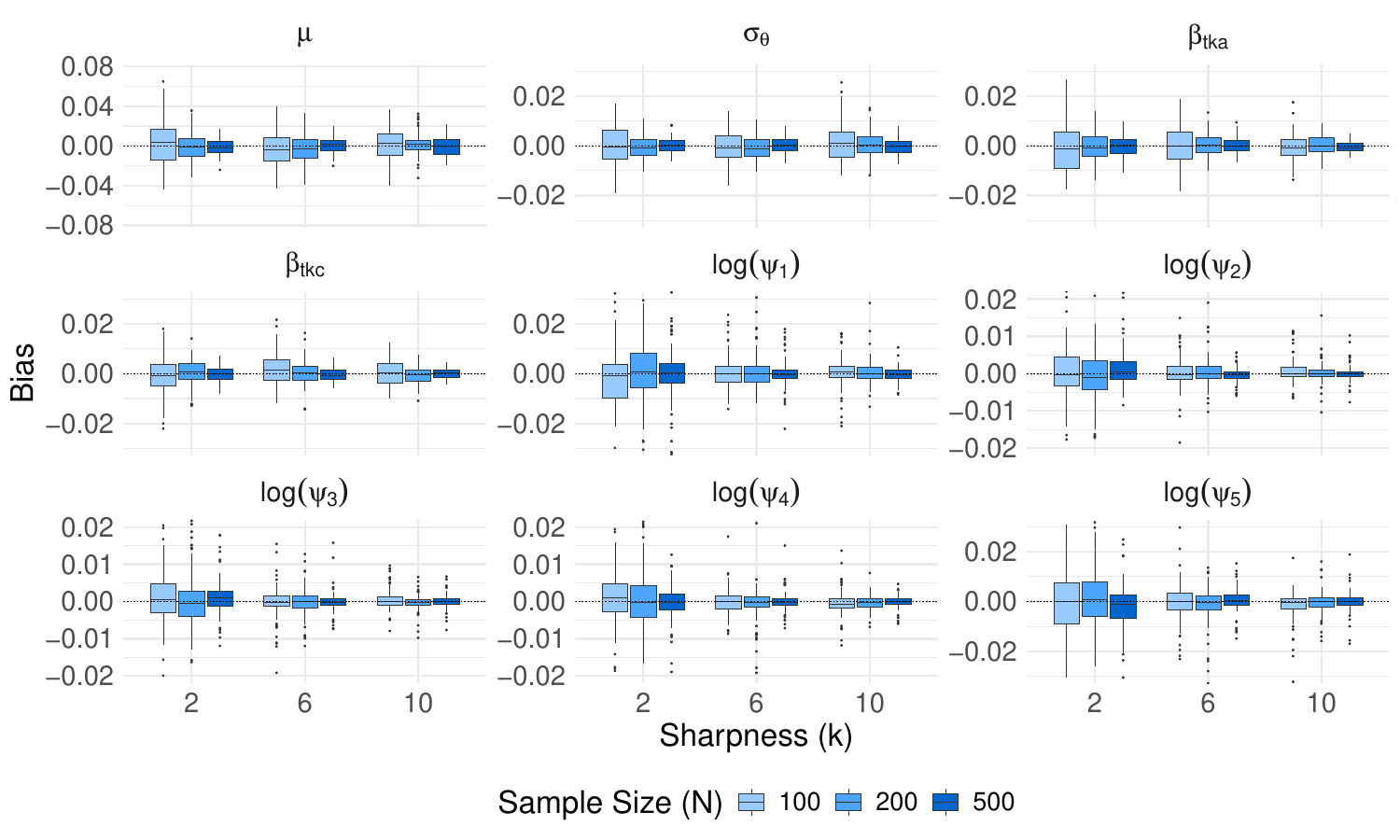}
    \caption{Bias plot of parameters for specimens in uniaxial tension configuration with $\sigma_\theta = 0.1$.}
    \label{fig:bias_plot_tension_0.1}
\end{figure}

\begin{figure}[h]
    \centering
    \includegraphics[width=\linewidth]{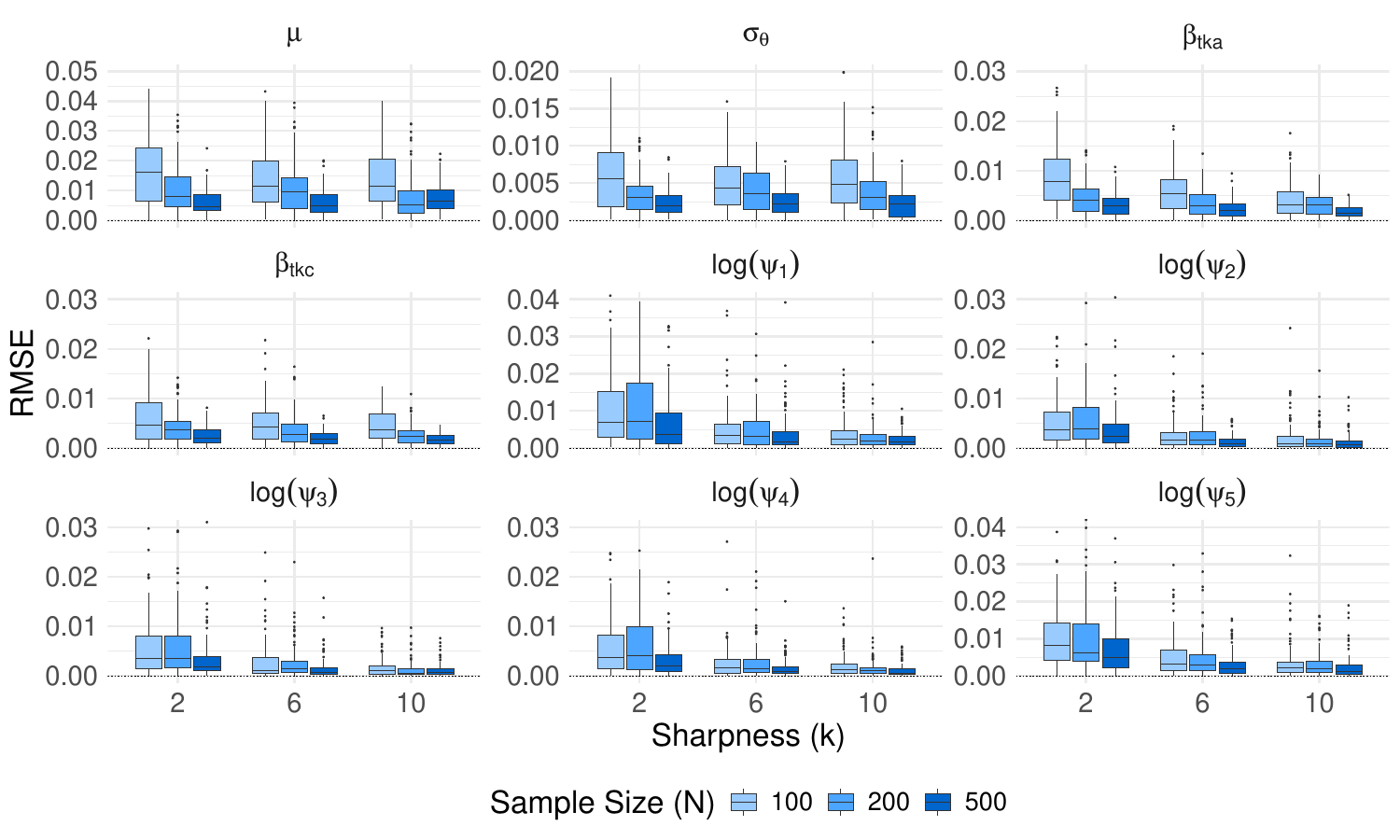}
    \caption{RMSE plot of parameters for specimens in uniaxial tension configuration with $\sigma_\theta = 0.1$.}
    \label{fig:rmse_plot_tension_0.1}
\end{figure}

\end{document}